\documentclass[twocolumn,10pt]{IEEEtran}
\usepackage{braket}

\input{def.tex}
\usepackage{dsfont}
\usepackage{minibox}
\DeclareSymbolFont{matha}{OML}{txmi}{m}{it}
\DeclareMathSymbol{\varv}{\mathord}{matha}{118}
\usepackage{eurosym}
\usepackage{hhline,physics}

\usepackage{multicol}
\usepackage{algorithm}
\usepackage{graphicx}
\usepackage{textcomp}
\usepackage{caption,pifont}
\usepackage[multiple]{footmisc}

\usepackage{algorithmic}

\newcommand{\trans}{^{\mathsf{T}}}
\newcommand{\herm}{^{\H}}
\makeatletter

\IEEEoverridecommandlockouts
\begin{document}
	\title{Two-Timescale Design for Active STAR-RIS Aided Massive MIMO Systems} 
	\author{Anastasios Papazafeiropoulos, Hanxiao Ge, Pandelis Kourtessis, Tharmalingam Ratnarajah, Symeon Chatzinotas, Symeon Papavassiliou \thanks{A. Papazafeiropoulos is with the Communications and Intelligent Systems Research Group, University of Hertfordshire, Hatfield AL10 9AB, U. K., and with SnT at the University of Luxembourg, Luxembourg. H. Ge and T. Ratnarajah are with the Institute for Digital Communications, The University of Edinburgh, Edinburgh, EH9 3FG, UK. P. Kourtessis is with the Communications and Intelligent Systems Research Group, University of Hertfordshire, Hatfield AL10 9AB, U. K. S. Chatzinotas is with the SnT at the University of Luxembourg, Luxembourg. 
			Symeon Papavassiliou is with the Institute of Communication and Computer Systems, School of Electrical and Computer Engineering,
			National Technical University of Athens, 15780 Zografou, Greece. 
			S. Chatzinotas   was supported by the National Research Fund, Luxembourg, under the project RISOTTI. E-mails: tapapazaf@gmail.com, \{hanxiao.ge, t.ratnarajah\}@ed.ac.uk p.kourtessis@herts.ac.uk, symeon.chatzinotas@uni.lu.}}
	\maketitle\vspace{-1.7cm}
	\begin{abstract}
		Simultaneously transmitting and reflecting \textcolor{black}{reconfigurable intelligent surface} (STAR-RIS) is a promising implementation of RIS-assisted systems that enables full-space coverage. However, STAR-RIS as well as conventional RIS suffer from the double-fading effect. Thus, in this paper, we propose the marriage of active RIS and STAR-RIS, denoted as ASTARS for massive multiple-input multiple-output (mMIMO) systems, and we focus on the energy splitting (ES) and mode switching (MS) 
		protocols. Compared to prior literature, we consider the impact of correlated fading, and we rely our analysis on the two timescale protocol, being dependent on statistical channel state information (CSI).  On this ground, we propose a channel estimation method for ASTARS with reduced
		overhead that accounts for its architecture. Next, we derive a \textcolor{black}{closed-form expression} for the achievable sum-rate for both types of users in the transmission and reflection regions in a unified approach with significant practical advantages such as reduced complexity and overhead, which result in a lower number of required iterations for convergence compared to an alternating optimization (AO) approach. Notably, we maximize simultaneously the amplitudes, the  phase shifts, and the active amplifying coefficients of the ASTARS by applying the projected gradient ascent method (PGAM).  Remarkably, the proposed optimization can be executed at every several coherence intervals that  reduces the  processing burden considerably. Simulations corroborate the analytical results, provide insight into the effects of  fundamental variables on the sum  achievable SE, and present the superiority of
		ASTARS  compared to passive STAR-RIS for a practical number of surface elements.		
	\end{abstract}
	
	\section{Introduction}
	Reconfigurable intelligent surfaces (RISs), being an emerging cost-efficient candidate for next-generation communications systems, have attracted significant interest from both academia and industry \cite{DiRenzo2020,Pan2021}. An RIS is a reconfigurable engineered surface, which consists of a large number of passive and low-cost scattering elements. In particular, a conventional RIS does not include any active radio frequency (RF) chains, digital signal processing units, and power amplifiers. A RIS achieves nearly passive beamforming gain by an adjustment of the amplitudes and phase shifts of the impinging waves through a controller towards constructively strengthening the desired signal or deconstructively weakening the interference signals. The deployment of RIS brings an increase to network capacity, reduction to the transmit power, improvement of transmission  reliability, and enlargement of wireless coverage \cite{Wu2019,DiRenzo2020}.
	
	Many previous works  relied on the assumption of instantaneous channel state information (CSI) availability, which is estimated in each channel coherence interval \cite{Wu2019,Huang2019}. However, this approach faces two challenges. The first one concerns the overhead for instantaneous CSI acquisition, which is proportional to the number of RIS elements. Given that an RIS generally includes a large number of reflecting elements \cite{Bjoernson2019b}, this overhead becomes prohibitively high. The second challenge concerns the requirement for calculation for the optimal beamforming coefficients of the RIS that also have to be sent to the RIS controller through a dedicated feedback link. The frequent execution of the information feedback and beamforming calculation induce high feedback overhead, computational complexity, and energy consumption. Hence, this approach is not recommended for scenarios with short coherence time such as  cases of high mobility.
	
	These challenges were initially addressed by Han et al. \cite{Han2019}, which proposed a two-time-scale design, which was further analysed in recent research works \cite{Zhao2020,Kammoun2020,Papazafeiropoulos2021,Abrardo2021,Chen2021,Papazafeiropoulos2022a,Chen2022,Chen2022a,Wu2023,Papazafeiropoulos2023}. According to this scheme, the design of the base station (BS) beamforming is based on the instantaneous CSI of the aggregated channel that is independent of the number of RIS elements, and thus, reduces significantly the channel estimation overhead. Moreover, the optimization of the RIS is based on statistical CSI in terms of large-scale statistics, which varies much slower than in instantaneous CSI. Hence, RIS-assisted design based on statistical CSI can significantly reduce the feedback overhead, the energy consumption, and the computational complexity because the passive beamforming matrix (PBM) requires to be updated only when the CSI varies, which takes place in a much larger scale than the instantaneous CSI.

	Generally, the performance of RIS is degraded by the double fading effect, which describes the total path loss of the cascaded transmitter-RIS-receiver link, being the product of the path losses of the transmitter RIS and RIS-receiver link \cite{Long2021,Zhang2022}. This path loss is usually an order of magnitude larger than that of the direct link. According to \cite{You2022}, the RIS should be located close to the BS or the users to enhance the communication performance. Also, it has been shown that conventional full-duplex (FD) amplify-and-forward (AF) relaying can outperform RIS if not large surfaces are deployed \cite{Bjoernson2019,Ntontin2020}. In particular, in \cite{Bjoernson2019}, it was shown that very large RISs are required to outperform decode-and-forward (DF) relaying, while in \cite{Ntontin2020}, a performance comparison took place in terms of achievable spectral efficiency (SE) and energy efficiency, and it was demonstrated that RIS-assisted systems surpass relay-aided networks only in the case of sufficiently large RIS. However, a large number of reflecting elements will result in high overhead as well as high power consumption to perform the reflect circuit operations \cite{Wang2020,You2020}. As a result, purely passive RIS leads to limited application scenarios.
	
	Fortunately, active RIS has been recently proposed \cite{Long2021,Zhang2022} as a solution to overcoming the previously mentioned weaknesses of passive RIS through the application of the reflected signal with low-cost hardware. Note that conventional AF relaying demands power-hungry RF chains, while in active RIS-assisted systems, the RIS reflects the signal directly in an FD way with low-power reflection-type amplifiers. 	Hence, in \cite{Long2021,Zhang2022}, it was shown that active-RIS-assisted systems achieve higher data rates than the passive RIS due to the amplification gain at the RIS. In addition, in \cite{You2021}, the downlink/uplink communications were considered and optimized separately the RIS placement for maximum of the RIS with a passive and active RIS. According to simulations, active RIS should be deployed closer to the receiver side while the passive RIS could be located at the transmitter or the receiver side. Real field experiment results in \cite{Zhang2022} demonstrated that active RIS can obtain a capacity gain, which is several tens of times higher than that of passive RIS. In the case of a typical indoor scenario, in active RIS,  a signal-to-noise ratio (SNR) 30-40 dB higher than that of the passive RIS was achieved, while a 25 dB capacity gain was observed in a single-user communication system \cite{Long2021}. Other recent studies focusing on system performance improvement in terms of the bit error rate, security, and energy efficiency were reported in \cite{Liu2021,Tasci2022}.
	
	Another major limitation of the conventional RIS design is its limited coverage in the half space in front of the surface. In \cite{Liu2021a}, a novel simultaneously transmitting and reflecting RIS (STAR-RIS) implementation was proposed to achieve $ 360^{\circ} $ coverage and fill this gap. According to this structure, the incident electro-magnetic (EM) waves are split into two portions. One portion penetrates through the surface and propagates behind the surface while the other portion is reflected back. In \cite{Xu2021}, a general hardware model was presented, and it was shown that STAR-RIS can result in a higher diversity order than the passive RIS. In \cite{Mu2021}, three operating protocols of STAR-RIS were proposed and studied the joint power splitting for power minimisation in a two-user network. In \cite{Papazafeiropoulos2023a} and \cite{Papazafeiropoulos2023b}, the sum rate and max-min rate of STAR-RIS assisted massive multiple-input multiple-output (mMIMO) systems were obtained for Rayleigh correlated channels by optimizing  the amplitudes and the phase shifts of the surface simultaneously. In the latter work, the impact of hardware impairments was also investigated. In \cite{Papazafeiropoulos2023c}, the introduction of STAR-RIS into cell-free mMIMO systems was considered and its presence was evaluated. Notably, the idea of STAR-RIS coincides with another recently presented structure denoted as intelligent omni-surface (IOS) \cite{Zhang2020b,Zeng2021,Zhang2022a}, where simultaneous reflection and transmission take place by using different hardware techniques. For example, in \cite{Zeng2021}, it was shown that IOS can enhance the SE for single-input and single-output (SISO) single-user systems.

	Given that pure passive RIS such as STAR-RIS is degraded by the severe double fading effect while conventional RIS has limited half-space coverage, in \cite{Liu2022,Chen2022,Ma2022}, a novel architecture has been proposed that enhances EM waves while providing simultaneously full space coverage. This architecture, met with different names such as double-faced active (DFA) RIS, will be denoted in this work as active star surface (ASTARS). The ASTARS can tune the power distribution between the transmission and reflection directions and adjust the amplifying coefficient. 
	
	\textit{Contributions}:	Against the above background, in this paper, we propose a downlink two time-scale design of an ASTARS-assisted mMIMO system subject to the general circumstances of correlated Rayleigh fading and imperfect CSI. In particular, the low complexity maximum ratio transmission (MRT) precoder is used for the active beamforming at the BS with regard to the instantaneous overall CSI while the amplitudes, the phase shifts, and the active amplifying coefficients (AACs) are optimised simultaneously based on statistical CSI. First, we obtain the estimated overall channel, and then, we derive a closed form of the achievable rate in terms of only large-scale statistics. Next, we optimise simultaneously the amplitudes, the phase shifts, and the AACs of the ASTARS towards sum-rate maximisation. The main contributions of this paper are summarised as follows.
	\begin{itemize}
		\item Unlike  \cite{Liu2022,Chen2022,Ma2022}, which rely on instantaneous CSI and do not account for correlated fading,  we rely our design on the two timescale protocol \cite{Zhao2020, Papazafeiropoulos2021}, and we consider the impact of spatial correlation. Compared to \cite{Papazafeiropoulos2023a}, we have assumed that the STAR-RIS is active to mitigate the double-fading effect. Hence, in addition to the amplitudes and phase shifts of the STAR-RIS elements, we optimize simultaneously the AACs.
		\item A unified analysis, concerning  the
		channel estimation for UEs
		lied in any of the $ t $ and $ r $ regions,  results in a more attractive design of the considered problem. The unified analysis continues for the data transmission phase. In particular, for the channel estimation, we employ the linear minimum mean square error  (LMMSE) method, and we derive a closed-form expression for the estimated channel   while taking into consideration the impact of spatial correlation and AACs. Similarly, by using MRT precoding, we derive a closed-form expression for the achievable sum SE.
		\item Since the ensuing analysis relies on statistical CSI in the context of large-scale  statistics that change at every several coherence intervals,  the  optimization of the ASTARS can be performed at these intervals that decreases considerably the overhead. In contract, prior works, elaborating on instantaneous CSI, are not feasible in practice because of  large overheads since small-scale statistics change at each coherence interval.\footnote{	It is worthwhile to mention that this characteristic is quite crucial for ASTARS 			applications that have thrice the number of optimization
			variables with respect to the  conventional RIS. } 
		\item 		We develop the problem of deriving simultaneously the optimal  amplitudes, phase shifts, and AACs of the ASTARS that optimize the achievable sum SE. Specifically, despite the non-convexity of the accounted problem and the coupling among the optimization parameters, we propose an iterative efficient
		method built on the projected gradient ascent method (PGAM), where all the ASTARS optimization parameters 	are updated at the same time per iteration. To the best of
		our knowledge, the simultaneous 		optimization of all the parameters of ASTARS takes place for the first time. This is a major  contribution because most  works in the litarature usually optimize just the phase shifts or optimize
		both the amplitudes and the phase shifts by using an alternating optimization (AO) approach, which requires an increased number of iterations. 
		\item Simulations  are illustrated  to provide insight into the effects of system variables and to demonstrate
		the outperformance of ASTARS over STAR-RIS under specific conditions. For 		instance, it is shown that this superiority is exhibited for a certain number of surface elements, while, in the case of a very large number of elements, the power consumption of ASTARS is prohibitively high and results in performance loss.
	\end{itemize}

	\textit{Paper Outline}: The structure  of this paper is arranged as follows. Section~\ref{System} describes the system model of a mMIMO system aided by a combination of active and STAR-RIS architectures. Section~\ref{ChannelEstimation} presents the channel estimation  and the downlink data transmission. Also, the  achievable sum SE for the ASTARS and its special cases are provided. Section \ref{Problem1} presents the simultaneous optimization of all three parameters concerning the ASTARS.  Section~\ref{Numerical} presents the numerical results, and Section~\ref{Conclusion} provides the conclusion.
	
	\textit{Notation}: Matrices  and  vectors are represented by boldface upper  and lower case symbols, respectively. The notations $(\cdot)^\T$, $(\cdot)^\H$, and $\tr\!\left( {\cdot} \right)$ denote the transpose, Hermitian transpose, and trace operators, respectively. Also, the symbols $ \arg\left(\cdot\right) $,  $\EE\left[\cdot\right]$, and $\odot$ denote the argument function,  the expectation operator, and the  entry-wise multiplication, respectively. The notation  $\diag\left(\bA\right) $ represents a vector with elements equal to the  diagonal elements of $ \bA $, and the notation  $\diag\left(\bx\right) $ represents a diagonal  matrix whose elements are $ \bx $. The notation  $\bb \sim \cC\cN{(\b0,\mathbf{\Sigma})}$ represents a circularly symmetric complex Gaussian vector with zero mean and a  covariance matrix $\mathbf{\Sigma}$. 
	
	\begin{figure}[!h]
		\begin{center}
			\includegraphics[width=0.9\linewidth]{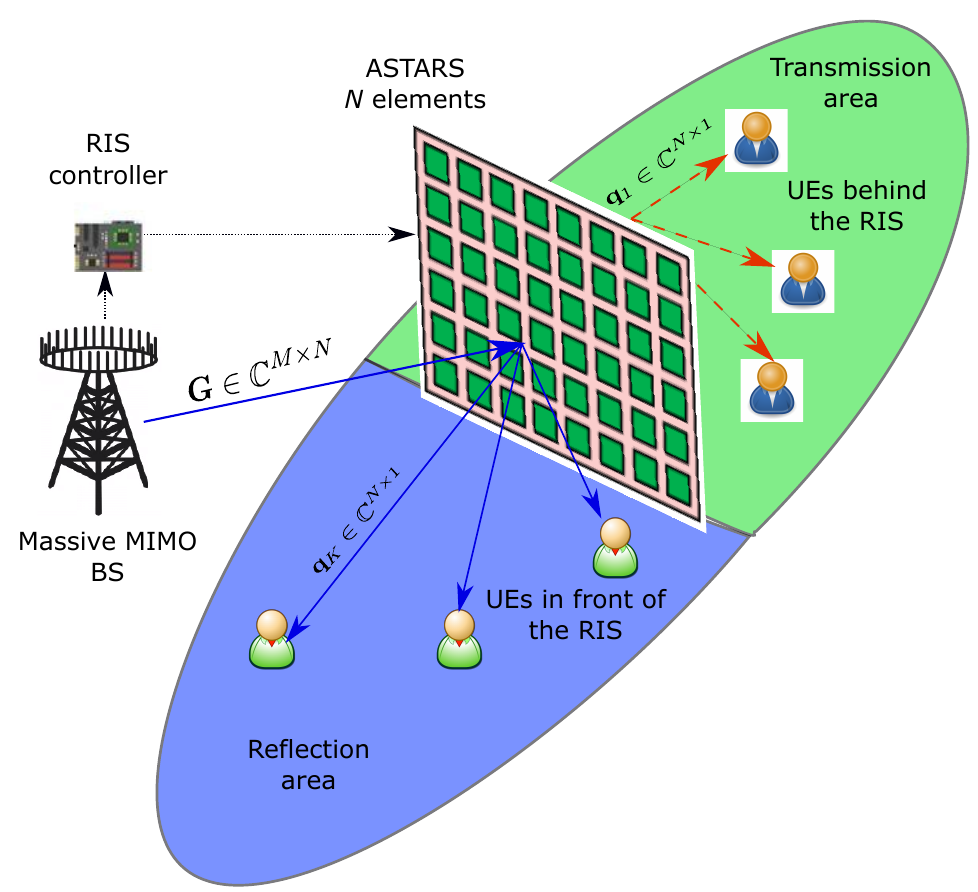}
			\caption{{ A mMIMO ASTARS-aided system with numerous UEs at each side of the surface.  }}
			\label{Fig1}
		\end{center}
	\end{figure} 
	\section{System Model}\label{System}
	We account for a downlink multi-user multiple-input single-output (MISO) system aided by a STAR-RIS, as shown in Fig. \ref{Fig1}. In particular, a BS, equipped with $ M $ antennas  serves simultaneously $ K $ single-antenna UEs that are located around the surface. Without loss of generality, $\mathcal{K}_{t}=\{1,\ldots,K_{t} \} $ UEs lie inside the  transmission region $ (t) $, i.e., behind the surface, while $ \mathcal{K}_{r}=\{1,\ldots,K_{r} \} $ UEs lie inside the reflection region $ (r) $,  where $ K_{t}+K_{r}=K $. To provide a unified model for both kinds of UEs and simplify the notation,  we denote the set $ \mathcal{W} = \{w_{1}, w_{2}, ..., w_{K}\}  $ to define the  operation mode  of the RIS for each of the $ K $ UEs. In particular, we denote $ w_{k} = t $, if $ k\in   \mathcal{K}_{t}$, i.e., if the $ k $th UE is located behind the STAR-RIS. Otherwise, if $ k\in   \mathcal{K}_{r}$, i.e., the $ k $th UE is facing the STAR-RIS, we denote  $ w_{k} = r $. Additionally to the active STAR-RIS, we take into account the existence of direct paths  between the UEs and the BS.
	
	\subsection{ASTARS}
	Conventional active RIS appears as a promising solution to the fundamental bottleneck of RISs, being the "multiplicative fading" effect. Compared to a passive RIS, an active RIS can not only reflect the incident signals by adjusting the phase	shifts but also amplify them. An integrated active 	 amplifier to the RIS can achieve this goal, where an active component can be for example an asymmetric current mirror or a   current-inverting converter \cite{Zhang2022}.
	
	Regarding the STAR-RIS, it is implemented by a uniform planar array (UPA) composed of a set of $\mathcal{N}=\{1, \ldots, N \}$ elements. The number of surface elements is $ N= N_{\mathrm{x}} \times N_{\mathrm{y}} $   with  $ N_{\mathrm{x}} $ and 	 $ N_{\mathrm{y}} $ being the number of horizontally and vertically active elements, respectively.  The STAR-RIS provides the  configuration of     reflected ($ r $) and the transmitted ($ t $) signals in terms of the   phase and amplitude parameters, i.e.,   $ \phi_{n}^{w_{k}} \in [0,2\pi)$ and $ {\beta_{n}^{w_{k}}}\in [0,1] $. Herein, we propose the marriage of active RIS with STAR-RIS, denoted as ASTARS. According to the proposed model, we denote  $ r_{n}=( a_{n}{\beta_{n}^{r}}e^{j \phi_{n}^{r}})s_{n} $  and  $ t_{n} =(a_{n} {\beta_{n}^{t}}e^{j \phi_{n}^{t}})s_{n}$ the reflected and transmitted 	signal by the $ n $th  element of the STAR-RIS, respectively. Note that   $ \al_{n}>1 $, $ n \in \mathcal{N} $ is the  AAC of the $ n $th ASTARS element,  According to this model, we can choose $ \phi_{n}^{t} $ and $ \phi_{n}^{r} $   independently, while  the amplitudes can be chosen  based on the relationship 
	\begin{align}
		(\beta_{n}^{t})^{2}+(\beta_{n}^{r})^{2}=1,  \forall n \in \mathcal{N}.
	\end{align}

	Below, we focus on the two main protocols of STAR-RIS, which were presented in \cite{Mu2021}, and are denoted as  energy splitting (ES) and mode switching (MS) . Their main points follow.
	\subsubsection{ES protocol}  This protocol is the most general since it assumes that all surface elements communicate simultaneously  with all UEs, located in $ r $ and $ t $ regions. For the $ k
	$th  UE, the PBM is given as $ \bPhi_{w_{k}}^{\mathrm{ES}}=\diag( {\beta_{1}^{w_{k}}}\theta_{1}^{w_{k}}, \ldots,  {\beta_{N}^{w_{k}}}\theta_{N}^{w_{k}}) \in \mathbb{C}^{N\times N}$, where $ \beta_{n}^{w_{k}} \ge 0 $, $ 		(\beta_{n}^{t})^{2}+(\beta_{n}^{r})^{2}=1 $, and $ |\theta_{n}^{w_{k}}|=1, \forall n \in \mathcal{N} $.
	
	\subsubsection{MS protocol} A subcase of the ES protocol is the MS protocol. Specifically,  the amplitude coefficients for reflection  and transmission take only  binary values.  	 In this case, 	we separate the RIS elements into  $ N_{t} $ and $ N_{r}=N-N_{t} $ elements, which serve UEs in the  corresponding regions. When $  k \in \mathcal{K}_{t} $ or $  k \in \mathcal{K}_{r} $, the PBM 
	is written as $ \bPhi_{w_{k}}^{\mathrm{MS}}=\diag( {\beta_{1}^{w_{k}}} \theta_{1}^{w_{k}}, \ldots,  {\beta_{N}^{w_k}}\theta_{N}^{w_{k}}) \in \mathbb{C}^{N\times N}$, where $ \beta_{n}^{w_{k}}\in \{0,1\}$, $ 	(\beta_{n}^{t})^{2}+(\beta_{n}^{r})^{2}=1 $, and $ |\theta_{i}^{w_{k}}|=1, \forall n \in \mathcal{N} $. Although the surface elements cannot achieve the full-dimension transmission and reflection beamforming gain in this case, MS  induces lower computational complexity for the PBM design. 	To simplify notation, 	 we denote $ \theta_{i}^{w_{k}}=e^{j\phi_{i}^{w_{k}}} $. Also, we neglect the superscript $ \mathcal{X}=\{ES,MS\} $ unless otherwise specified.
	
	The amplified signal by the $ N $-element ASTARS can be expressed as
	\begin{align}
		\by=\bA\bPhi_{w_{k}}\bs+\bA\bPhi_{w_{k}}\bv+\bn,\label{receivedsignal}
	\end{align}
	where $ \bs =[s_{1},\ldots, s_{N}]^{\T}$  is the incident signal vector at the ASTARS, $ \bv\sim \mathcal{CN}(\b0, \sigma_{v}^{2}\Id_{N})$, and $ \bA= \diag(\sqrt{\al_{1}}, \ldots,  \sqrt{ \al_{N}}) $ is the amplitude coefficient matrix at the ASTARS. In \eqref{receivedsignal}, the first term describes the desired signal, and the second and third terms correspond to  the noise introduced by the ASTARS  that can be classified as dynamic and static. The former comes from the noise induced and amplified by the reflection-type amplifier, while the latter is unrelated to  $ \bPhi_{w_{k}} $ and comes from the patch and the phase-shift circuit \cite{Bousquet2012}. Below, we omit $ \bn $ since it  is usually negligible compared to  the dynamic noise $ \bA\bPhi_{w_{k}}\bv $ \cite{Zhang2022}.

	\subsection{Channel Model}\label{ChannelModel} 
	We consider  narrowband quasi-static block fading channels. Each block has a duration of $\tau_{\mathrm{c}}$ channel uses. According to \cite{Zhi2022}, it is preferable to deploy the surface close to  the UEs in  a mMIMO system. Hence, we assume that the RIS-UEs links are  line-of-sight (LoS).\footnote{To satisfy this assumption, the surface could be placed in a tall building in the proximity of the UEs. Under this setting, the power of the LoS components is much stronger than that of the  non-line-of-sight (NLoS) components  for the RIS-UEs channels. Hence, for analytical tractability, we can assume that these links are purely LoS.} Except from these links, all other channels are liable to Rayleigh fading with spatial correlation since the surface, being close to the UEs, is expected to be far away from the BS with no LoS paths.\footnote{The scenario, including LoS paths, i.e., Rician fading, will be considered as the topic of future work.} Specifically, the channel  between the BS and the ASTARS, the channel between the ASTARS and UE $ k $ that may be located at either side, and the direct link between the BS  and UE  $ k $ are described as
	\begin{align}
		\bG&=\sqrt{\tilde{ \beta}_{g}}\bR_{\mathrm{BS}}^{1/2}\bD\bR_{\mathrm{RIS}}^{1/2},\label{eq2}\\
		\bq_{k}&=\sqrt{\tilde{ \beta}_{k}}\ba_{N}\left(\phi_{k,r}^{\al},\phi_{k,r}^{e}\right)\nn\\
		&\triangleq \tilde{\ba}_{N_{x}}(\sin\phi_{k,r}^{\al}\sin \phi_{k,r}^{e})\otimes \tilde{\ba}_{N_{y}}(\cos  \phi_{k,r}^{e}),\label{los}\\
		\bd_{k}&=\sqrt{\bar{ \beta}_{k}}\bR_{\mathrm{BS}}^{1/2}\bar{\bc}_{k},	
	\end{align}
	where  $ \bR_{\mathrm{BS}} \in \mathbb{C}^{M \times M} $ and $ \bR_{\mathrm{RIS}} \in \mathbb{C}^{N \times N} $  are the deterministic Hermitian-symmetric positive semi-definite correlation matrices at the BS and the RIS respectively.\footnote{ $\bR_{\mathrm{BS}} \in \mathbb{C}^{M \times M} $ and $ \bR_{\mathrm{RIS}} \in \mathbb{C}^{N \times N} $ are assumed known by the network. In particular, $  \bR_{\mathrm{BS}} $ can be modeled   as in \cite{Hoydis2013}, and $ \bR_{\mathrm{RIS}} $ can be modeled as in \cite{Bjoernson2020}. 
		\textcolor{black}{Another way for the practical calculation of  the covariance matrices for both $ \bR_{\mathrm{RIS}} $ and $ \bR_{\mathrm{BS}}$, despite the fact that $ \bR_{\mathrm{RIS}} $ is passive follows. Specifically, it can be observed that the expressions of these covariance matrices depend on the distances between the RIS elements and the BS antennas, respectively, as well as the angles between them.  The angles can be calculated when the locations are given, while the distances are known from the construction of the RIS and the BS. Thus, the covariance matrices can be considered to be known. Similarly, the  path loss coefficient of the separate channel between RIS and UE $ k $, depending on the distance, is assumed known since the corresponding distance is known. }	} In addition, $\tilde{ \beta}_{g} $, $ \bar{ \beta}_{k} $, and $\tilde{ \beta}_{k} $  describe the path-losses  of the BS-ASTARS, BS-UE $ k $, and ASTARS-UE $ k $ links in $ t $ or $r  $ region, respectively. The random vectors $ \mathrm{vec}(\bD)\sim \mathcal{CN}\left(\b0,\Id_{MN}\right) $ and $ \bc_{k} \sim \mathcal{CN}\left(\b0,\Id_{N}\right) $ describe the  fast-fading components. 
	In \eqref{los}, we have considered the two-dimensional uniform rectangular array (URA) to describe the LoS channels \cite{Wu2019}, where $ \phi_{k,r}^{\al},\phi_{k,r}^{e} $ are the azimuth and elevation angles of arrival (AoA) from  the RIS to UE $ k $.  Note that $ \tilde{\ba}_{b}(c)\triangleq [1,\ldots, e^{j 2\pi \frac{d}{\lambda }(b-1)c}]^{\T} $, $ b \in \{N_{\mathrm{x}}, N_{\mathrm{y}}\} $ with $ \ba_{N}^{\H}\left(\phi_{k,r}^{\al},\phi_{k,r}^{e}\right) \ba_{N}\left(\phi_{k,r}^{\al},\phi_{k,r}^{e}\right)=N$.	In the sequel, we denote 	 $ \ba_{N}\left(\phi_{k,r}^{\al},\phi_{k,r}^{e}\right) $ simply by $ \ba_{N}$ to simplify the notation.

	The aggregated channel vector for UE $ k $ is $ \bh_{k}=\bd_{k}+ \bG\bA\bPhi_{w_{k}} \bq_{k} $. Given the PBM and since $ \bG $ is LoS, we have $ \bh_{k} \sim \mathcal{CN}\left(\b0,\bR_{k}\right) $, i.e., it  is a zero-mean complex Gaussian vector with  variance  $ \bR_{k}=\EE\{\bh_{k}\bh_{k}^{\H}\} $, which is written as
	\begin{align}
		\bR_{k}=\bar{ \beta}_{k}\bR_{\mathrm{BS}}+\hat{\beta}_{k}\tr(\bR_{\mathrm{RIS}}\bA \bPhi_{w_{k}} \ba_{N}\ba_{N}^{\H}  \bPhi_{w_{k}}^{\H}\bA)\bR_{\mathrm{BS}.\label{cov1}}
	\end{align}
	In \eqref{cov1}, we have considered that   $ \bG $ and $ \bq_{k} $ are independent,  $ \EE\{	\bq_{k}	\bq_{k}^{\H}\} =\tilde{ \beta}_{k} \bR_{\mathrm{RIS}}$. Also, we have considered that  $ \EE\{\bV \bU\bV^{\H}\} =\tr (\bU) \Id_{M}$, where $\bU  $ is a deterministic square matrix, and $ \bV $ is any matrix with independent and identically distributed (i.i.d.) entries having zero mean and unit variance. Also, we have denoted $\hat{\beta}_{k}= \tilde{ \beta}_{g}\tilde{ \beta}_{k} $.

	\begin{remark}
		It is worthwhile to mention that if $\bR_{\mathrm{RIS}} =\Id_{N} $, the covariance $ \bR_{k} $ is independent of  the phase shifts but depends on the amplitudes  and the AACs. Specifically, the variance of aggregated channel is written as $ \bR_{k}\triangleq\bar{\bR}_{k}=\bar{ \beta}_{k}\Id_{M}+\hat{\beta}_{k}\sum_{i=1}^{N}\al_{i}(\beta_{i}^{w_{k}})^{2}\Id_{M} $, which does not depend on the phase shifts. As will be seen below, the achievable SE cannot be optimized in terms of phase shifts in this case. This is a characteristic of works relying on statistical CSI \cite{Zhao2020,Papazafeiropoulos2021,Papazafeiropoulos2023}. Fortunately, in practice, correlated fading occurs. Thus, surface optimization  can take place with respect to both phase shifts and amplitudes.
	\end{remark}
	
	\section{ASTARS Channel Estimation and Downlink Transmission}\label{ChannelEstimation}
	In this section, we first elaborate on the channel estimation method performed at the BS to obtain the estimate $ \hat{\bh}_{k} $ of the end-to-end channel $ \bh_{k} $ between UE $ k $ and the BS when the BS has only statistical CSI. Next, based on the provided channel estimate, we present the downlink transmission to UE $ k $, $ \forall k $ for the scenario where UEs are aware of only channel statistics. The downlink received signal will be later used to derive the corresponding achievable rate.
	
	\subsection{Channel Estimation Under Statistical CSI}
	Herein, instead of optimizing the individual surface-assisted channels, we estimate the end-to-end channel $ \bh_{k} $ for a given $ \bPhi_{w_{k}} $ to use it for beamforming at the BS as mentioned in the Introduction. The reasons are: i) To save prohibitive training overhead coming from estimating all individual surface-assisted channels as the number of surface elements increases due to the large dimension of $ M\times N $ (both $ M $ and $ N $ could be large) in the BS-RIS channel; ii) Optimization  of the PBM based on instantaneous CSI at the pace of fast fading channels increases significantly the complexity of the system, while optimization based on statistical CSI in terms of large-scale statistics suffices.
	
	To estimate the channel coefficients $ \bh_{k} $ in the uplink training phase of length $ \tau \ge K $ symbols with the aim to reduce the pilot overhead, the BS employs the minimum mean-square error (MMSE) estimator. All $ K $ UEs, located on both  regions, transmit simultaneously orthogonal pilot sequences of length $ \tau \le K $ symbols to the BS. Hence, the minimum pilot sequence length is $ \tau= K$ as in mMIMO that means it does not depend on $ M $ and $ N $. In particular, $\bx_{k}=[x_{k,1}, \ldots, x_{k,\tau}]^{\H}\in \mathbb{C}^{\tau\times 1} $ denotes the pilot sequence of UE $ k  $, which can be located in any region. Notably, we have $ \bx_{k}^{\H}\bx_{l}=0~\forall k\ne l$ and $ \bx_{k}^{\H}\bx_{k}= \tau p$ joules with $ p =|x_{k,i}|^{2} ,~\forall k,i$, which means that all UEs use the same average  power during this phase.
	
	For the whole uplink training period, the received signal by the BS is written as
	\textcolor{black}{\begin{align}
		\bY^{\tr}=\sum_{i=1}^{K}(\bh_{i}\bx_{i}^{\H} +\bDelta_{i})+
		\bZ^{\tr}\label{train1}.
	\end{align}}Herein,  $ \bZ^{\tr} \in \mathbb{C}^{M \times \tau} $ is the  AWGN matrix with independent columns. Each column is distributed as $ \mathcal{CN}\left(\b0,\sigma^2\Id_{M}\right)$. Also,  $ \bDelta_{i} \in \mathbb{C}^{M \times \tau} $ is the additive dynamic noise introduced by the active RIS. \textcolor{black}{Similar to $  \bZ^{\tr} $, each column of $ \bDelta_{i} $ equals $ \bG \bA \bPhi_{w_{i}}\bv $.} Notably, in \eqref{train1}, the UEs from both $ t $ or $ r $ regions contribute.
	
	By multiplying \eqref{train1} with  $ \frac{\bx_{k}}{\sqrt{\tau p}} $, we obtain 
	\begin{align}
		\br_{k}=\bh_{k}+\frac{1}{\sqrt{\tau p}}\sum_{i=1}^{K}\deltav_{i}+\frac{\bz_{k}}{\sqrt{\tau p}},\label{train2}
	\end{align}
	where $ \deltav_{i}=\bDelta_{i}\bx_{i}=\bG \bA \bPhi_{w_{i}}\bv $ and $ \bz_{k}=\bZ^{\tr} \bx_{k}$. Since $ \bh_{k} $ is Gaussian distributed, we can  apply the well-known MMSE estimator
	to obtain its channel estimate.\footnote{The MMSE estimation has been widely utilized in conventional mMIMO systems \cite{Marzetta2016,Ngo2013}. However, this method has not been used before for STAR-RIS.
		This is one of the advantages of this work. i.e., to manage 
		in a smart way to introduce the standard channel estimation for  multiple-user SIMO to STAR-RIS by making the expression for the channel estimation "look" the same for both types of UEs belonging to different areas of the surface. Nevertheless, the channel estimation   is different since   $  {\bR}_{k} $ is different, for users in each region. Moreover, the treatment of the channel seems to be standard but it is not since it concerns the aggregated channel vector.}
	
	The proposed estimation has the same dimension as the UE-BS channel matrix, and the length of the pilot sequences is $ \tau_{\mathrm{c}}\ge K $. In other words, the  computational complexity and overhead are lower compared to other methods that estimate the $ MN $ individual channels in RIS-aided communications \cite{Wang2020,Nadeem2020}.
	
	\begin{lemma}\label{PropositionDirectChannel}
		The  estimated  end-to-end channel by MMSE   is written as
		\begin{align}
			\hat{\bh}_{k}=\bR_{k}\bQ_{k} \br_{k},\label{estim1}
		\end{align}
		where $ \bQ_{k}\!=\! \left(\bR_{k} +\frac{{\tilde{ \beta}_{g}}\sigma_{v}^{2}}{{\tau p}}\sum_{j=1}^{K}\tr(\bR_{\mathrm{RIS}}\bA^2\bC_{w_{j}})\bR_{\mathrm{BS}} + \frac{\sigma^2}{ \tau p }\Id_{M}\right)^{\!-1}$ with $ \bC_{w_{j}}=\diag((\beta_{1}^{w_{j}})^{2}, \ldots,  (\beta_{N}^{w_{j}})^{2}) $, and $ \br_{k}$ is the noisy channel given by \eqref{train2}. The estimation error is denoted as 
		$ \tilde{\bh}_{k}= \bh_{k}-\hat{\bh}_{k}$, where the error $ \tilde{\bh}_{k} $ is independent of the estimate $ \hat{\bh}_{k} $. Also, the MSE matrix is obtained as
		\begin{align}
			\mathrm{\textbf{MSE}}_{k}=\bE_{k}\triangleq\bR_{k}-	\bPsi_{k}\label{mse}
		\end{align}
		with 
		\begin{align}
			\bPsi_{k}\triangleq\EE\left\{\hat{\bh}_{k}	\hat{\bh}_{k}^{\H}\right\}=\bR_{k}\bQ_{k}\bR_{k}.\label{var1}
		\end{align} 
	\end{lemma}
	\begin{proof}
		See Appendix~\ref{lem1}.	
	\end{proof}
	\begin{figure*}
	\begin{align}
		\mathrm{\textbf{MSE}}_{k}&\!=\!\left[\!\!\left(\!\!\left(\!\!\bar{ \beta}_{k}+\hat{\beta}_{k}\lambda_{RIS}\lambda_{BS}\sum_{i=1}^{N}\al_{i}(\beta_{i}^{w_{k}})^{2}\right)\!\!\Id_{M}\!\!\right)^{\!\!-1}\!\!+\!\!\left(\!\!{\tilde{ \beta}_{g}}\sigma_{v}^{2}\lambda_{RIS}\lambda_{BS}\sum_{j=1}^{K}\sum_{i=1}^{N}\al_{i}(\beta_{i}^{w_{j}})^{2}
		+ \frac{\sigma^2}{ \tau p }\right)^{\!\!-1}\!\Id_{M}\!\right]^{-1}\nn\\
		&	=\frac{1}{\frac{1}{\bar{ \beta}_{k}+\hat{\beta}_{k}\lambda_{RIS}\lambda_{BS}\sum_{i=1}^{N}\al_{i}(\beta_{i}^{w_{k}})^{2}+{\tilde{ \beta}_{g}}\sigma_{v}^{2}\lambda_{RIS}\lambda_{BS}\sum_{j=1}^{K}\sum_{i=1}^{N}\al_{i}(\beta_{i}^{w_{j}})^{2}}+\frac{\tau p}{ \sigma^2}}\Id_{M}.\label{CorMSE}
	\end{align}
	\hrulefill
\end{figure*}

	\begin{corollary}\label{cor1}
		Consider the special case of  $ \bR_{BS}=\lambda_{BS}\Id $ and $ \bR=\lambda_{RIS}\Id $ (independent Rayleigh fading). Application of \cite[Eq. 15.67]{Kay} to  the error covariance matrix in \eqref{mse} provides \eqref{CorMSE}, 	where the last equation is obtained after some simple manipulations. From \eqref{CorMSE}, it can be observed that the MSE is an increasing function of $ M $, $ N $, $ \sigma_{v}^{2} $, the various path-losses, $ \sigma^2 $, the amplifying coefficients and the amplitude coefficients of the ASTARS elements. The increase in the estimation error due to $ N $ comes from the creation of more communication paths between the BS and the UEs. In the case of the absence of the ASTARS and independent Rayleigh fading, the estimated channel $ 	\hat{\bh}_{k} $ MSE reduce to  $ 	\hat{\bh}_{k}=\frac{\bar{ \beta}_{k}}{\bar{ \beta}_{k}+\frac{\sigma^2}{ \tau p}}\br_{k} $ and $ 	\mathrm{\textbf{MSE}}_{k}=\frac{\bar{ \beta}_{k}\frac{\sigma^2}{ \tau p}}{\bar{ \beta}_{k}+\frac{\sigma^2}{ \tau p}}\Id_{M} $, which are the same as for conventional mMIMO systems \cite{Ngo2013}.
	\end{corollary}
	In this direction, given  the significance of channel estimation for mMIMO systems,  we provide the NMSE of the estimate of $ \bh_{k} $ as
	\begin{align}
		\mathrm{NMSE}_{k}
		&=\frac{\mathrm{\textbf{MSE}}_{k}}{\tr(\bR_{k})}\\
		&=1-\frac{\tr(\bPsi_{k})}{\tr(\bR_{k})},\label{nmse1}
	\end{align}
	where $\mathrm{\textbf{MSE}}_{k}= \tr(\EE[(\hat{\bh}_{k}-{\bh}_{k})(\hat{\bh}_{k}-{\bh}_{k})^{\H}]) $, and the NMSE  lies in the range $ [0,1] $ \cite{massivemimobook}.
	
	\begin{corollary}\label{cor2}
		Under the conditions of independent Rayleigh fading as in Corollary \ref{cor2}, we obtain \eqref{CorNMSE},
		\begin{figure*}
			\begin{align}
				\mathrm{NMSE}_{k}&=\frac{\frac{ \sigma^2}{\tau p}}{\bar{ \beta}_{k}+\hat{\beta}_{k}\lambda_{RIS}\lambda_{BS}\sum_{i=1}^{N}\al_{i}(\beta_{i}^{w_{k}})^{2}+{\tilde{ \beta}_{g}}\sigma_{v}^{2}\lambda_{RIS}\lambda_{BS}\sum_{k=1}^{K}\sum_{i=1}^{N}\al_{i}(\beta_{i}^{w_{}})^{2}+\frac{ \sigma^2}{\tau p}}.\label{CorNMSE}		
			\end{align}
			\hrulefill
		\end{figure*}
		where, in the large $ N $ regime,  high
		pilot power-to-noise ratio regime, and  low pilot power-to-noise ratio regime, the asymptotic $ \mathrm{NMSE}_{k} $ is given by
		\begin{align}
			\lim_{N\to \infty}	\mathrm{NMSE}_{k}\to 0,\label{NMSEN}\\
			\lim_{\frac{ \sigma^2}{\tau p}\to 0}	\mathrm{NMSE}_{k}\to 0,\label{NMSE0}\\
			\lim_{\frac{ \sigma^2}{\tau p}\to \infty}	\mathrm{NMSE}_{k}\to 1.\label{NMSE1}
		\end{align}
		
		Also, if the power $ p $ is is scaled proportionally to $ p=E_{\mathrm{u}}/N $ with $ E_{\mathrm{u}}  $ being a constant. When $ N \to \infty $, $ \mathrm{NMSE}_{k} $ converges to a limit less than one as in \eqref{A_tilde_general}
		\begin{figure*}
			\begin{align}
				\lim_{p=\frac{ E_{\mathrm{u}}}{N}, N\to \infty}\!\!\!\!\!\!	\mathrm{NMSE}_{k}\!\to\! \frac{{ \sigma^2}}{\textcolor{black}{\frac{\tau E_{\mathrm{u}}}{N}}\hat{\beta}_{k}\lambda_{RIS}\lambda_{BS}\!\sum_{i=1}^{N}\!\al_{i}(\beta_{i}^{w_{k}})^{2}\!+\!\textcolor{black}{\frac{\tau E_{\mathrm{u}}}{N}}{\tilde{ \beta}_{g}}\sigma_{v}^{2}\lambda_{RIS}\lambda_{BS}\!\sum_{j=1}^{K}\sum_{i=1}^{N}\!\al_{i}(\beta_{i}^{w_{j}})^{2}\!+\!{ \sigma^2}}<1.\label{NMSEinfty}
			\end{align}
			\hrulefill
		\end{figure*}
	\end{corollary}
	
	\begin{proof}
		When $ N\to\infty $		or $ \frac{ \sigma^2}{\tau p}\to \infty $, by focusing on the dominant terms, we observe  that we result in \eqref{NMSEN} and \eqref{NMSE1}, respectively. In the case of $ \frac{ \sigma^2}{\tau p}\to 0 $, the numerator of \eqref{CorNMSE} becomes zero, which gives \eqref{NMSE0}.
	\end{proof}
	
	Generally, a method for  reducing the NMSE in conventional mMIMO systems is to increase the length of the pilot
	sequence $ \tau $. However, Corollary \ref{cor2} shows that an increase in $ N $ can have the same effect. In other words, an increase in the number of surface elements increases not only the rate but also reduces the NMSE. Also, \eqref{NMSEinfty} shows when the surface includes a large number of elements, the NMSE converges to a  value of less than one despite the use of low pilot powers. \textcolor{black}{Actually, this scaling law shows that the power does not increase but reduces while increasing the size of the ASTARS.}
	
	Notably, for $ N\to 0 $ in \eqref{CorMSE} and \eqref{CorNMSE}, we obtain the expressions corresponding to conventional mMIMO systems. We observe that the NMSE  of STAR-RIS-assisted mMIMO systems is better than the NMSE of conventional mMIMO, while the MSE is worse. The justification is that the surface includes $ N $ additional paths, while the pilot length is increased. As a result, the estimation error is increased. The normalized error is decreased due to better channel gains with $ N $.
	
	The following corollary sheds further light on the impact of increasing $ N $ in the channel estimation.
	
	\begin{remark}\label{cor3}
		Although the channel estimation is perfect when $ \tau \to \infty $ since $ \tilde{\bh}_{k}\to 0 $ and $ 	\hat{\bh}_{k}\to 	{\bh}_{k} $, when $ N \to \infty $, we result in impefect CSI. In particular, we obtain $ \tilde{\bh}_{k}\ne 0 $, $ \mathrm{\textbf{MSE}}_{k}=\frac{\sigma^2}{ \tau p}\Id_{M} $, and $ 	\mathrm{NMSE}_{k}\to 0 $.
	\end{remark}
	\begin{proof}
		These results are obtained easily from the corresponding expressions.
	\end{proof}

	\subsection{Downlink Transmission}
	During the downlink data transmission phase, the ASTARS transmits and reflects signals simultaneously from the UEs to the BS. Also, in this paper, we account for linear precoding, and we assume that CSI is known by means of relevant channel estimation techniques. Thus, UE $ k $ receives
	\begin{align}
		r_{k}&=\bd_{k}^{\H}\bx+\bq_{k}^{\H}\bPhi_{w_{k}}^{\H}\bA\bG^{\H}\bx+ \bq_{k}^{\H}\bPhi_{w_{k}}^{\H}\bA\bv+z_{k}\label{DLreceivedSignal1}\\
		&=	{\bh}^\H_{k}\bx+ \bq_{k}^{\H}\bPhi_{w_{k}}^{\H}\bA\bv+z_{k}.\label{DLreceivedSignal2}
	\end{align}
	Herein,   $\bx=\sqrt{\frac{\lambda P}{K}} \sum_{i=1}^{K}\bff_{i}l_{i}$ describes the signal vector transmitted by the BS,  $ P$ denotes the downlink transmit power. Also,   $\bff_{i} \in \bbC^{M \times 1}$ expresses the linear precoding vector for UE $ i $, and 
	$ l_{i} $ is     the  data symbol with $ \EE\{|l_{i}|^{2}\}=1 $. Notably, as commonly assumed in the mMIMO literature \cite{Hoydis2013}, we have adopted all UEs share equal power. The parameter $ \lambda $  is obtained such that $ \EE[\mathbf{x}^{\H}\mathbf{x}]=p $, which gives $ \lambda=\frac{K}{\EE\{\tr\bF\bF^{\H}\}} $,
	with $ \bF=[\bff_{1}, \ldots, \bff_{K}] \in\mathbb{C}^{M \times K}$. Moreover, $z_{k} \sim \cC\cN(0,1)$ denotes the AWGN at UE $k$.

	\subsection{Achievable Spectral Efficiency}
	By 	using  the use-and-then-forget bound \cite{Bjoernson2017}, we  derive a lower bound on downlink average SE in bps/Hz  as
	\begin{align}
		\mathrm{SE}	= b \sum_{k=1}^{K}	\mathrm{SE}_{k}
	\end{align}
	with
	\begin{align}
		\mathrm{SE}_{k}=\log_{2}\left ( 1+\gamma_{k}\right)\!,\label{LowerBound}
	\end{align}
	where $ b=\frac{\tau_{\mathrm{c}}-\tau}{\tau_{\mathrm{c}}} $ describes    the percentage of samples per coherence block  for downlink data transmission, while $ \gamma_{k}$ describes the   signal-to-interference-plus-noise ratio (SINR) given by
	\begin{align}
		\gamma_{k}=	\frac{	S_{k}}{{I}_{k}}\label{gamma1}
	\end{align}
	with 	
	\begin{align}
		{{S}}_{k}&=|\EE\{\bh_{k}^{\H}\bff_{k}\}|^{2}, \label{Sig} \\
		{{I}}_{k}&=\EE\big\{ \big| {\bh}_{k}^{\H}\hat{\bh}_{k}-\EE\big\{
		{\bh}_{k}^{\H}\hat{\bh}_{k}\big\}\big|^{2}\big\}\!+\!\left({\sigma_{v}^{2}}{ }\big\|\bq_{k}^{\H}\bPhi_{w_{k}}^{\H}\bA\big\|_{2}^{2}\!+\!\sigma^{2}\right)\frac{K}{P \lambda}\nn\\
		&+\!\sum_{i\ne k}^{K}|\EE\{\bh^\H_{k}\bff_{i}\}|^{2}\!.\label{Int}\end{align}

\begin{theorem}\label{TheoremRate}
	Let correlated Rayleigh fading, imperfect CSI, and a given PBM, then the downlink achievable SINR of UE $k$   in ASTARS-aided  mMIMO system with MRT precoding  is given by
	\begin{equation}
		\gamma_{k} =	\frac{S_{k}}{	I_{k}},\label{gammaSINR}
	\end{equation}
	where
	\begin{align}
		{S}_{k}&=\tr^{2}\left(\bPsi_{k}\right)\!,\label{Num1}\\
		I_{k}&=\sum_{i =1}^{K}\tr\!\left(\bR_{k}\bPsi_{i} \right)\nn\\
		&+\Big(\sigma_{v}^{2}\tilde{ \beta}_{k}\sum_{i=1}^{N}\al_{i}(\beta_{i}^{w_{k}})^{2}+\sigma^{2}\Big)\frac{K}{ p}\sum_{i=1}^{K}\tr(\bPsi_{i}),\label{Den1}
	\end{align}
	where  $ \bC_{w_{k}}=\diag( {\beta_{1}^{w_{k}}} , \ldots,  {\beta_{N}^{w_k}}) $.
\end{theorem}
\begin{proof}
	See Appendix~\ref{th1}.	
\end{proof}

Thus, according to Theorem \ref{TheoremRate}, we obtain \eqref{LowerBound1}.
\begin{figure*}
	\begin{align}
		\mathrm{SE}_{k}=\log_{2}\!\!\left (\!\! 1\!+\!\frac{\tr^{2}\left(\bPsi_{k}\right)}{\sum_{i =1}^{K}\tr\!\left(\bR_{k}\bPsi_{i} \right)\!+\!\left(\!\sigma_{v}^{2}\tilde{ \beta}_{k}\sum_{i=1}^{N}\al_{i}(\beta_{i}^{w_{k}})^{2}+\sigma^{2}\!\right)\frac{K}{ p}\sum_{i=1}^{K}\tr(\bPsi_{i})}\!\!\right)\!.\label{LowerBound1}
	\end{align}
	\hrulefill
\end{figure*}
As can be seen, the achievable SE is written in terms of only  large-scale statistics since we have followed an approach based on statistical CSI. Thus, we can update the parameters of the ASTARS only at every several coherence intervals that reduces considerably the computational complexity and overhead.

Before the ASTARS optimization, we focus on the benefits of the ASTARS and shed light on its potential with respect to a conventional RIS.

\begin{corollary}
	When the ASTARS is reduced to conventional STAR-RIS, the SE in \eqref{LowerBound1} becomes
	\begin{align}
		\!\!	\mathrm{SE}_{k}=\log_{2}\!\!\left (\!\! 1\!+\!\frac{\tr^{2}\left(\tilde{\bPsi}_{k}\right)}{\sum_{i =1}^{K}\tr\!\left(\tilde{\bR}_{k}\tilde{\bPsi}_{i} \right)+\frac{K\sigma^{2}}{ p}\sum_{i=1}^{K}\tr(\tilde{\bPsi}_{i})}\!\!\right)\!,\label{LowerBound2}
	\end{align}
\end{corollary}
where 
\begin{align}
	\tilde{\bR}_{k}&=\bar{ \beta}_{k}\bR_{\mathrm{BS}}+\hat{\beta}_{k}\tr(\bR_{\mathrm{RIS}} \bPhi_{w_{k}} \ba_{N}\ba_{N}^{\H}  \bPhi_{w_{k}}^{\H})\bR_{\mathrm{BS}\label{cov2}},\\
	\tilde{\bPsi}_{k}&=\tilde{\bR}_{k}\tilde{\bQ}_{k}\tilde{\bR}_{k}
\end{align}
with $ \tilde{\bQ}_{k}\!=\! \left(\tilde{\bR}_{k} + \frac{\sigma^2}{ \tau p }\Id_{M}\right)^{\!-1}$.

\subsection{Benefits of ASTARS}

To demonstrate the benefits of integrating an ASTARS, we resort to the special of independent Rayleigh fading. First, we discuss the scenario assuming no surface, and then, we compare it with the scenario having an ASTARS.
\subsubsection{mMIMO without any  surface}
In the case of no surface, the following corollary holds.
\begin{corollary}
	When the ASTARS is switched off, which means that $ \hat{\beta}_{k}=\tilde{ \beta}_{k}=0 $ for all UEs, the SE in \eqref{LowerBound1} becomes
	\begin{align}
		\mathrm{SE}_{k}=\log_{2}\!\!\left (\!\! 1\!+\!\frac{\tr^{2}\left(\bar{\bPsi}_{k}\right)}{\sum_{i =1}^{K}\tr\!\left(\bar{\bR}_{k}\bar{\bPsi}_{i} \right)\!+\!\frac{K\sigma^{2}}{ p}\sum_{i=1}^{K}\tr(\bar{\bPsi}_{i})}\!\!\right)\!,\label{LowerBound3}
	\end{align}
\end{corollary}
where
\begin{align}
	\tilde{\bR}_{k}&=\bar{ \beta}_{k}\bR_{\mathrm{BS}},\\	\bar{\bPsi}_{k}&=\bar{ \beta}_{k}^{2}\bR_{\mathrm{BS}}^{2}\left(\bar{ \beta}_{k}\bR_{\mathrm{BS}} + \frac{\sigma^2}{ \tau p }\Id_{M}\right)^{\!-1}.
\end{align}

\begin{corollary}
	If we further simplify the SE in \eqref{LowerBound3} by assuming uncorrelated channels, which means $ \bR_{\mathrm{BS}}=\Id_{M} $, \eqref{LowerBound3} becomes 
	\begin{align}
		\mathrm{SE}_{k}=\log_{2}\!\!\left (\!\! 1\!+\!\frac{\left(\frac{\bar{ \beta}_{k}^{2}}{\bar{ \beta}_{k}+\frac{\sigma^2}{ \tau p }}\right)^{2}M}{\sum_{i =1}^{K}\frac{\bar{ \beta}_{k}\bar{ \beta}_{i}^{2}}{\bar{ \beta}_{i}+\frac{\sigma^2}{ \tau p }}\!+\!\frac{K\sigma^{2}}{ p}\sum_{i=1}^{K}\frac{\bar{ \beta}_{i}^{2}}{\bar{ \beta}_{i}+\frac{\sigma^2}{ \tau p }}}\!\!\right)\!,\label{LowerBound4}
	\end{align}
	

	As can be seen, the complexity order of  \eqref{LowerBound4} is $ \mathcal{O}(\log_{2}(M)) $. Also, it can be shown that by scaling down the power proportionally to $ p=E_{d}/ \sqrt{M} $, where $ E_{d} $ is a constant, the SE can maintain a non-zero value as the number of antennas $ M\to \infty $, which is given by
	\begin{align}
		\lim_{M\to \infty}
		\overline{\mathrm{SE}}_{k}\to \hat{\tau}\log_{2}\!\left(1+\frac{\tau E_{d}^{2}\bar{ \beta}_{k}^{2}\sigma^{-4}}{{K\sigma^{2}}\sum_{i=1}^{K}{\bar{ \beta}_{i}}}\right)\label{rate33}\!.
	\end{align}
\end{corollary}

The SE in \eqref{LowerBound4} and the power scaling law in \eqref{rate33}, which coincides with \cite[Eq. 79]{Kong2015} will act as benchmarks that will allow us to shed light on the properties introduced by the ASTARS.
\subsubsection{ASTARS-assisted mMIMO with independent fading}
The complexity order of \eqref{LowerBound1} is $ \mathcal{O}(\log_{2}(MN)) $ since the trace in the denominator is of order $ N $. The source of this gain is the $ N $ additional paths coming from the RIS, which allows more signals to be received by the BS. Also, a higher capacity can be obtained with respect to the scenario of the mMIMO system without a RIS, which has a scaling law of $ \mathcal{O}(\log_{2}(M)) $. Moreover, the scaling law $ \mathcal{O}(\log_{2}(MN)) $ signifies that the number of BS antennas can be reduced inversely proportional to the number of the ASTARS while keeping the SE fixed.
\begin{corollary}
	Under conditions of independent Rayleigh fading, the SE in \eqref{LowerBound1} gives \eqref{LowerBound2},
	\begin{figure*} 
		\begin{align}
			\mathrm{SE}_{k}=\log_{2}\!\!\left (\!\! 1\!+\!\frac{\left(\frac{\bar{	e}_{k}^{2}}{e_{k}}\right)^{2}M}{\sum_{i =1}^{K}\frac{\bar{	e}_{k}\bar{	e}_{i}^{2}}{e_{i}}\!+\!\left(\!\sigma_{v}^{2}\tilde{ \beta}_{k}\sum_{i=1}^{N}\al_{i}(\beta_{i}^{w_{k}})^{2}+\sigma^{2}\!\right)\frac{K}{ p}\sum_{i=1}^{K}\frac{\bar{	e}_{i}^{2}}{e_{i}}}\!\!\right)\!.\label{LowerBound2}
		\end{align}
		\hrulefill
	\end{figure*}
	where 
	\begin{align}
		\bar{	e}_{k}&=\bar{ \beta}_{k}+\hat{\beta}_{k}\sum_{i=1}^{N}\al_{i}(\beta_{i}^{w_{k}})^{2},\\
		e_{k}&=\bar{	e}_{k}+{\tilde{ \beta}_{g}}\sigma_{v}^{2}\sum_{i=1}^{N}\al_{i}(\beta_{i}^{w_{k}})^{2}+\frac{\sigma^2}{\tau p }.
	\end{align}
\end{corollary}
The numerator of  \eqref{LowerBound2} increases with $ N^{2} $ . 
Also, the denominator in \eqref{LowerBound2} grows without bound as $ N \to \infty$, i.e., it increases with $ N^{2} $. As a result, when $ N $ is large, the SE of ASTARS saturates.

Although an attractive characteristic of conventional mMIMO systems is their power scaling law with the number of antennas, i.e., the transmit power can be scaled down while increasing $ M $ as shown in \eqref{rate33} and as already shown in \cite{Kong2015}, herein, we present a new power scaling in terms of $ N $, and we collate it with \eqref{rate33}.
\begin{corollary}
	When the power is scaled proportionally to $ p=E_{d}/N^{2} $ with $ N \to \infty $, the achievable SE in \eqref{LowerBound2} can maintain a non-zero value given by \eqref{LowerBound3}.
	\begin{figure*}
		\begin{align}
			\mathrm{SE}_{k}=\log_{2}\!\!\left (\!\! 1\!+\!\frac{\left(\frac{\bar{	e}_{k}^{2}}{e_{k}}\right)^{2}M}{\sigma_{v}^{2}\tilde{ \beta}_{k}K\sum_{i=1}^{N}\al_{i}(\beta_{i}^{w_{k}})^{2}\sum_{i=1}^{K}\frac{(\hat{\beta}_{k}\sum_{i=1}^{N}\al_{i}(\beta_{i}^{w_{k}})^{2})^{2}}{\sigma^{2}_{i}}}\!\!\right)\!.\label{LowerBound3}
		\end{align}
		\hrulefill
	\end{figure*}
\end{corollary}

\begin{proof}
	We substitute $ p=\frac{E_{d}}{N^{2}} $ in \eqref{LowerBound2}. When $ N \to \infty $, $\frac{\bar{	e}_{k}^{2}}{e_{k}}\to =\tau \sigma^{-2} E_{d}/N^{2}(\hat{\beta}_{k}\sum_{i=1}^{N}\al_{i}(\beta_{i}^{w_{k}})^{2})^{2}  $.
\end{proof}


\section{Problem Formulation}\label{Problem1}
In this section, we formulate an optimization problem to obtain the PBM that maximizes the achievable sum SE given by Theorem \eqref{TheoremRate} based on statistical CSI.

\subsection{ASTARS Optimization for ES protocol}
It is crucial to find the PBM and the AAC, which maximize the sum SE  of  the ASTARS-assist systems  with correlated fading. Based on infinite-resolution phase shifters, we formulate the corresponding optimization problem. 

Taking into account that the ASTARS uses a large number of low-cost amplifiers, a per-element power constraint  for each amplifier takes place, which is given as
\begin{align}
	\EE\{\al_{n}|s_{n}+v_{n}|^{2}\}\le P_{n}, \forall m \in \mathcal{N}
\end{align}
where $ P_{n} $ is the  effective transmit radio power per ASTARS element. This element-wise power constraint can be further written as
\begin{align}
	\al_{n}\sum_{i=1}^{K}|\bff_{i}|^{2}+\al_{n}\sigma_{v}^{2}\le P_{n}, \forall m \in \mathcal{N}.
\end{align}

The total transmission of the ASTARS is determined by the constraint
\begin{align}
	\EE\{\|\by\|_{2}^{2}\}=\EE\{\|\bA(\bs+\bv)\|_{2}^{2}\}\le P_{R},
\end{align}
where $ P_{R} $ is the  total effective transmit  power of the ASTARS. This constraint can be further obtained as
\begin{align}
	\sum_{i=1}^{K}\|\bA\bff_{i}\|_{2}^{2}+\sigma_{v}^{2}\|\bA\|_{F}^{2}\le P_{R}.
\end{align}

The formulated problem reads
\begin{equation}
	\begin{IEEEeqnarraybox}[][c]{rl}
		(\mathcal{P}) \qquad \max_{\alv,\betv, \thetv}&\quad	\mathrm{SE}	=\frac{\tau_{\mathrm{c}}-\tau}{\tau_{\mathrm{c}}}\sum_{k=1}^{K}\overline{\mathrm{SE}}_{k}\\
		\mathrm{s.t}&\quad (\beta_{n}^{t})^{2}+(\beta_{n}^{r})^{2}=1,  \forall n \in \mathcal{N}\\
		&\quad\beta_{n}^{t}\ge 0, \beta_{n}^{r}\ge 0,~\forall n \in \mathcal{N}\\
		&\quad|\theta_{n}^{t}|=|\theta_{n}^{r}|=1, ~\forall n \in \mathcal{N}\\
		&\quad\al_{n}\sum_{i=1}^{K}|\bff_{i}|^{2}+\al_{n}\sigma_{v}^{2}\le P_{n}, \forall m \in \mathcal{N}\\
		&\quad	\sum_{i=1}^{K}\|\bA\bff_{i}\|_{2}^{2}+\sigma_{v}^{2}\|\bA\|_{F}^{2}\le P_{R},
	\end{IEEEeqnarraybox}\label{Maximization}
\end{equation}
where $\thetv=[(\thetv^{t})^{\T}, (\thetv^{r})^{\T}]^{\T} $ and $\betv=[(\betv^{t})^{\T}, (\betv^{r})^{\T}]^{\T} $  to achieve a compact description.
For ease of exposition, we define three sets: $ \Theta=\{\thetv\ |\ |\theta_{i}^{t}|=|\theta_{i}^{r}|=1,i \in \mathcal{N}\} $,  $ \mathcal{B}=\{\betv\ |\ (\beta_{i}^{t})^{2}+(\beta_{i}^{r})^{2}=1,\beta_{i}^{t}\geq0,\beta_{i}^{r}\geq0,i \in \mathcal{N}\} $, and $ \mathcal{A}=\{\alv|\al_{i},i \in \mathcal{N}\} $ which constitute the feasible set of \eqref{Maximization}. The ASTARS introduces new challenges. Specifically, we have two types of passive beamforming, which are transmission and reflection beamforming.  Note that they are also coupled  because of the energy conservation law. Moreover, we have the element-wise power, the total power constraints, and the AAC.

Not only problem $ (\mathcal{P}) $ is non-convex, but also a coupling regarding the amplitudes and the phase shifts for transmission and reflection appears. Also, it includes the AACs. We observe the simplicity of  the sets $\Theta$,  $ \mathcal{B}$, and $ \mathcal{A} $ since their projection operators can be written in closed form. For this reason, we are going to apply the  PGAM \cite[Ch. 2]{Bertsekas1999} for the  optimization of $\thetv$, $\betv$, and $\alv  $.\footnote{\textcolor{black}{Another methodology is the  manifold optimization that has been proved to be effective in handling non-convex unit-modulus 	constraints \cite{Absil2008,Chen2021a}. Its implementation and comparison with PGAM could be the topic of future work. } }  \textcolor{black}{The solution can not be claimed to be global but local since the problem $ 	(\mathcal{P}) $ is non convex.}

The suggested PGAM includes  the  iterations below. Specifically, we have
\begin{subequations}\label{mainiteration}\begin{align}
		\thetv^{n+1}&=P_{\Theta}(\thetv^{n}+\mu_{n}\nabla_{\thetv}\mathrm{SE}(\betv^{n},\thetv^{n}, \alv^{n})),\label{step1} \\ \betv^{n+1}&=P_{\mathcal{B}}(\betv^{n}+{\mu}_{n}\nabla_{\betv}\mathrm{SE}(\betv^{n},\thetv^{n},\alv^{n})),\label{step2}\\
		\alv^{n+1}&=P_{\mathcal{A}}(\alv^{n}+{\mu}_{n}\nabla_{\alv}\mathrm{SE}(\betv^{n},\thetv^{n},\alv^{n})),\label{step3} \end{align}
\end{subequations}
where the superscript corresponds to the iteration count, and $\mu_n$ denotes the step size for  all three parameters $\thetv$, $\betv$, and $\alv$. Moreover,  $P_{\Theta}(\cdot) $, $ P_{\mathcal{B}}(\cdot) $, and $ P_{\mathcal{A}}(\cdot) $ denote the projections onto $ \Theta $, $ \mathcal{B} $, and $ \mathcal{A} $, respectively. According to the algorithm, we move from the current iterate $(\betv^{n},\thetv^{n}, \alv^{n})$ along the gradient direction to increase the objective.

Notably, the choice of the step size  is important because it will determine the convergence of the  proposed PGAM.  Herein,  we employ the Armijo-Goldstein backtracking line search to obtain the parameter $\mu  $ at each iteration. In particular,  let $ L_n>0 $, and $ \kappa \in (0,1) $. In such case,  $  \mu_{n} $ in \eqref{mainiteration} can be obtained as $ \mu_{n} = L_{n}\kappa^{m_{n}} $ with $ m_{n} $ being the smallest nonnegative integer that satisfies
\begin{align}
	&\mathrm{SE}(\betv^{n+1},\thetv^{n+1}, \alv^{n+1})\nn\\
	&\geq	Q_{L_{n}\kappa^{m_{n}}}\big(\betv^{n}, \thetv^{n}, \alv^{n};\betv^{n+1},\thetv^{n+1}, \alv^{n+1}\big),
\end{align}
where
\begin{align}
	&	Q_{\mu}(\betv, \thetv,\alv;\bx,\by, \bz)=	\mathrm{SE}(\betv, \thetv, \alv)+\langle	\nabla_{\thetv}\mathrm{SE}(\betv, \thetv, \alv),\bx-\thetv\rangle\nn\\
	&-\frac{1}{\mu}\|\bx-\thetv\|^{2}_{2}+\langle\nabla_{\betv}\mathrm{SE}(\betv, \thetv, \alv),\by-\betv\rangle-\frac{1}{\mu}\|\by-\betv\|^{2}_{2}\nn\\
	&+\langle\nabla_{\alv}\mathrm{SE}(\betv, \thetv, \alv),\bz-\alv\rangle-\frac{1}{\mu}\|\bz-\alv\|^{2}_{2}.
\end{align}    

\textcolor{black}{Now, we provide the projection onto the sets  $ \Theta $, $ \mathcal{B} $, and $ \mathcal{A} $. Specifically,   for a given $\thetv\in \mathbb{C}^{2N\times 1}$  $P_{\Theta}(\thetv)$ is given by 
\begin{equation}
	P_{\Theta}(\thetv)=\thetv/|\thetv|=e^{j\angle\thetv}.
\end{equation}
In the case of the projection  $P_{ \mathcal{B} }(\betv)$, we observe that the constraint $(\beta_{i}^{t})^{2}+(\beta_{i}^{r})^{2}=1,\beta_{i}^{t}\geq0,\beta_{i}^{r}\geq0$  defines the first quadrant of the  unit circle, which makes the expression of the projection onto $\mathcal{B}$  rather complicated. By allowing  $\beta_{i}^{t}$ and $\beta_{i}^{r}$ to take negative value, we can make  $P_{\mathcal{B} }(\betv)$ more efficient without affecting the optimality of the proposed solution since  we can change the sign of both $\beta_i^{u}$ and $\theta_{i}^{u}$, $u\in{t,r}$ and still achieve the same objective. Hence,   we can project $\beta_{i}^{t}$ and $\beta_{i}^{r}$ onto the entire unit circle by writing $P_{ \mathcal{B} }(\betv)$  as
\begin{subequations}
	\begin{align}
		\left[\ensuremath{P_{\mathcal{B}}(}\boldsymbol{\beta})\right]_{i} & =\frac{{\beta}_{i}}{\sqrt{{\beta}_{i}^{2}+{\beta}_{i+N}^{2}}},i=1,2,\ldots,N\\
		\left[\ensuremath{P_{\mathcal{B}}(}\boldsymbol{\beta})\right]_{i+N} & =\frac{{\beta}_{i+N}}{\sqrt{{\beta}_{i}^{2}+{\beta}_{i+N}^{2}}}, i=1,2,\ldots,N.
	\end{align}
\end{subequations}}

\textcolor{black}{Regarding $ P_{\mathcal{A}}(\al_{n}) $, we have
\begin{align}
	P_{\mathcal{A}}(\al_{n}) = \left\{\begin{array}{ll}
		c_{1}, & \text{for } \al_{n}< c_{1}\\
		\al_{n}, & \text{for } c_{1}\leq \al_{n}\leq c_{2}\\
		c_{2}, & \text{for } \al_{n}>c_{2}
	\end{array}\right., 
\end{align}
where $ c_1=1 $ and $ c_{2}= P_{n}/(\sum_{i=1}^{K}|\bff_{i}|^{2}+\sigma_{v}^{2}) $.}

Algorithm \ref{Algoa1} includes a summary of the proposed PGAM. The gradients, which are complex-valued, are provided in the lemma below.
\begin{algorithm}[th]
	\caption{PGAM for the ASTARS Design\label{Algoa1}}
	\begin{algorithmic}[1]
		\STATE Input: $\thetv^{0},\betv^{0},\alv^{0}$, $\kappa\in(0,1)$,  $\mu_{1}>0$
		\STATE $n\gets1$
		\REPEAT
		\REPEAT \label{ls:start}
		\STATE $\thetv^{n+1}=P_{\Theta}(\thetv^{n}+\mu_{n}\nabla_{\thetv}\mathrm{SE}(\betv^{n},\thetv^{n}, \alv^{n}))$
		\STATE $\betv^{n+1}=P_{\mathcal{B}}(\betv^{n}+{\mu}_{n}\nabla_{\betv}\mathrm{SE}(\betv^{n},\thetv^{n},\alv^{n}))$
		\STATE 	$ \alv^{n+1}=P_{\mathcal{A}}(\alv^{n}+{\mu}_{n}\nabla_{\alv}\mathrm{SE}(\betv^{n},\thetv^{n},\alv^{n})) $
		\IF{ $\mathrm{SE}(\betv^{n+1},\thetv^{n+1}, \alv^{n+1})\le 	Q_{L_{n}\kappa^{m_{n}}}(\betv^{n}, \thetv^{n}, \alv^{n};\betv^{n+1},\thetv^{n+1}, \alv^{n+1})$}
		\STATE $\mu_{n}=\mu_{n}\kappa$
		\ENDIF
		\UNTIL{ $\mathrm{SE}(\betv^{n+1},\thetv^{n+1}, \alv^{n+1})>	Q_{L_{n}\kappa^{m_{n}}}(\betv^{n}, \thetv^{n}, \alv^{n};\betv^{n+1},\thetv^{n+1}, \alv^{n+1})$}\label{ls:end}
		\STATE $\mu_{n+1}\leftarrow\mu_{n}$
		\STATE $n\leftarrow n+1$
		\UNTIL{ convergence}
		\STATE Output: $\betv^{n+1},\thetv^{n+1},\alv^{n+1}$
	\end{algorithmic}
\end{algorithm}

\begin{lemma}\label{LemmaGradients}
	The complex gradients $ \nabla_{\thetv}\mathrm{SE} $,  $\nabla_{\betv}\mathrm{SE} $, and $\nabla_{\alv}\mathrm{SE} $ are given in closed forms by
	\begin{subequations}
		\begin{align}
			\nabla_{\thetv}\mathrm{SE} &=[\nabla_{\thetv^{t}}\mathrm{SE}^{\T}, \nabla_{\thetv^{r}}\mathrm{SE}^{\T}]^{\T},\\
			\nabla_{\thetv^{i}}\mathrm{SE}&=-\frac{\tau_{c}-\tau}{\tau_{c}\log_{2}(e)}\sum_{k=1}^{K}\frac{	{I_k}\nabla_{\thetv^{i}}{S_{k}}-S_{k}	\nabla_{\thetv^{i}}{{I_k}}}{(1+\gamma_{k}){I_k}^{2}} ,i=t,r
		\end{align}
	\end{subequations}
	where
	\begin{subequations}
		\begin{align}
			\nabla_{\thetv^{t}}S_{k}&=\begin{cases}
				\nu_{k}\diag\bigl(\mathbf{C}_{t}\diag(\boldsymbol{{\beta}}^{t})\bigr) & w_{k}=t\\
				0 & w_{k}=r
			\end{cases}\label{derivtheta_t},\\
			\nabla_{\thetv^{r}}S_{k}&=\begin{cases}
				\nu_{k}\diag\bigl(\mathbf{C}_{r}\diag(\boldsymbol{{\beta}}^{r})\bigr) & w_{k}=r\\
				0 & w_{k}=t
			\end{cases}\label{derivtheta_r},\\
			\nabla_{\thetv^{t}}{I_k} &=\diag\bigl(\tilde{\mathbf{A}}_{kt}\diag(\boldsymbol{{\beta}}^{t})\bigr)\label{derivtheta_t_Ik}\\
			\nabla_{\thetv^{r}{I}_k} &=\diag\bigl(\tilde{\mathbf{A}}_{kr}\diag(\boldsymbol{\beta}^{r})\bigr)\label{derivtheta_r_Ik}
		\end{align}
	\end{subequations}
	with  $ \nu_{k}=2\hat{\beta}_{k}\tr\left(\bPsi_{k}\right)\tr((\bQ_{k}\bR_{k}+\bR_{k}\bQ_{k}-\bQ_{k}\bR_{k}\bQ_{k})\bR_{\mathrm{BS}}) $,   $\mathbf{C}_{w_k}=\bA\mathbf{R}_{\mathrm{RIS}}\bA\bPhi_{w_k} \ba_{N}\ba_{N}^{\H}$
	for $w_k\in\{t,r\}$,
	\begin{equation}
		\tilde{\mathbf{A}}_{ku}=\begin{cases}
			\bigl(\bar{\nu}_{k}+\sum\nolimits _{i\in\mathcal{K}_{u}}^{K}\tilde{\nu}_{ki}\bigr)\mathbf{C}_{u} & w_{k}=u\\
			\sum\nolimits _{i\in\mathcal{K}_{u}}^{K}\tilde{\nu}_{ki}\mathbf{C}_{u} & w_{k}\neq u
		\end{cases},\label{A_tilde_general}
	\end{equation}
	$u\in\{t,r\}$, $\bar{\nu}_{k}=\hat{\beta}_{k}\tr\bigl({\boldsymbol{\Psi}}_{k}\mathbf{R}_{\mathrm{BS}}\bigr)$,
	$\tilde{\nu}_{ki}=\hat{\beta}_{k}\tr\bigl(\tilde{\mathbf{R}}_{ki}\mathbf{R}_{\mathrm{BS}}\bigr)$, $ \bPsi=\sum_{i=1}^{K}\bPsi_{i}$, $\tilde{\mathbf{R}}_{ki}=\mathbf{Q}_{i}\mathbf{R}_{i}\bar{\mathbf{R}}_{k}-\mathbf{Q}_{i}\mathbf{R}_{i}\bar{\mathbf{R}}_{k}\mathbf{R}_{i}\mathbf{Q}_{i}+\bar{\mathbf{R}}_{k}\mathbf{R}_{i}\mathbf{Q}_{i}$, and $\bar{\mathbf{R}}_{k}=\mathbf{R}_{k}+\Big(\sigma_{v}^{2}\tilde{ \beta}_{k}\sum_{i=1}^{N}\al_{i}(\beta_{i}^{w_{k}})^{2}+\sigma^{2}\Big)\frac{K}{ p}\mathbf{I}_{M}$.

	Similarly, the  gradient $\nabla_{\betv}\mathrm{SE} $ is given by
	\begin{subequations}\label{eq:deriv:wholebeta}
		\begin{align}
			\nabla_{\betv}\mathrm{SE} &=[\nabla_{\betv^{t}}\mathrm{SE}^{\T}, \nabla_{\betv^{r}}\mathrm{SE}^{\T}]^{\T},\\
			\nabla_{\betv^{i}}\mathrm{SE}&=-\frac{\tau_{c}-\tau}{\tau_{c}\log_{2}(e)}\sum_{k=1}^{K}\frac{	{I_{k}}\nabla_{\betv^{i}}{S_{k}}-S_{k}	\nabla_{\betv^{i}}{{I_{k}}}}{(1+\gamma_{k}){I_{k}}^{2}},i=t,r   \label{gradbetat:final}
		\end{align}
	\end{subequations} 
	where
	\begin{subequations}
		\begin{align}
			\nabla_{\betv^{t}}S_{k}&=\begin{cases}
				2\nu_k\Re\bigl\{\diag\bigl(\mathbf{C}_{k}\herm\diag(\btheta^{t})\bigr)\bigr\} & w_{k}=t\\
				0 & w_{k}=r
			\end{cases}\label{derivbeta_t},\\
			\nabla_{\betv^{r}}S_{k}&=\begin{cases}
				2\nu_k\Re\bigl\{\diag\bigl(\mathbf{C}_{k}\herm\diag(\btheta^{r})\bigr)\bigr\} & w_{k}=r\\
				0 & w_{k}=t
			\end{cases}\label{derivbeta_r},\\
			\nabla_{\betv^{t}}{I_k} &=2\Re\bigl\{\diag\bigl(\tilde{\mathbf{A}}_{kt}\herm\diag(\boldsymbol{\btheta}^{t})\bigr)\bigr\}\\
			\nabla_{\betv^{r}}{I_k} &=2\Re\bigl\{\diag\bigl(\tilde{\mathbf{A}}_{kr}\herm\diag(\boldsymbol{\btheta}^{r})\bigr)\bigr\}.
		\end{align}
	\end{subequations}
	
	In the case of the gradient $ \nabla_{\alv}\mathrm{SE} $, we have
	\begin{align}
		\nabla_{\alv}\mathrm{SE}&=-\frac{\tau_{c}-\tau}{\tau_{c}\log_{2}(e)}\sum_{k=1}^{K}\frac{	{I_{k}}\nabla_{\alv}{S_{k}}-S_{k}	\nabla_{\alv}{{I_{k}}}}{(1+\gamma_{k}){I_{k}}^{2}},
	\end{align}
	with 
	\begin{align}
		\nabla_{\alv}S_{k} &= 2 \nu_k\Re\bigl\{\diag\bigl(\bD\bigr)\},\\
		\nabla_{\boldsymbol{\beta}^{t}}{I_{k}}&=2\Re\bigl\{\diag\bigl(\tilde{\mathbf{A}}_{kt}\herm\diag(\boldsymbol{\btheta}^{t})\bigr)\bigr\}.
	\end{align}
\end{lemma}
\begin{proof}
	Please see Appendix~\ref{lem2}.	
\end{proof}

\subsection{Complexity Analysis}
Herein, we present the complexity analysis per iteration of the proposed algorithm  based on the big-$ \mathcal{O} $ notation. Notably, this analysis is quite meaningful because the current work concerns large $M$ and $N$.

The complexity depends on the complex multiplications in  the objective and gradient value of Algorithm \ref{Algoa1}. First, we focus on the computation of $ \bR_{k} $ in \eqref{cov1}, which depends on the computation of $ \tr(\bR_{\mathrm{RIS}}\bA \bPhi_{w_{k}} \ba_{N}\ba_{N}^{\H}  \bPhi_{w_{k}}^{\H}\bA) $. Since $ \bPhi_{w_{k}}^{\H} $ is diagonal, $\bR_{\mathrm{RIS}} \bPhi_{w_{k}}$ requires $N^2$ complex multiplications. Then, the multiplication with $ \bR_{\mathrm{BS}} $ adds $ M^{2} $ to the complexity, i.e., the complexity of $  \bR_{k} $ is  $O(M^{2}+N^2)$. Next, $ 	\bPsi_{k}=\bR_{k}\bQ_{k}\bR_{k} $, where $\bQ_{k}  $ is an inverse matrix,  has  a complexity  $O(M^3)$. Application of eigenvalue decomposition (EVD) as in \cite{Papazafeiropoulos2023a} results in $O(M^2)$, which is the compexity of $ S_{k} $ in \eqref{Num1}. Similarly, the complexity of $ I_{k} $ in \eqref{Den1} requires  $O(K(M^2+N^2))$ multiplications. In the case of the gradients,  to obtain $\nabla_{\thetv^{t}}S_{k}$ in \eqref{derivtheta_t}, we have to calculate $ \nu_{k} $, which requires  $O(M^2+N^2)$ multiplications. Next, $ \diag\bigl(\mathbf{C}_{t}\diag(\boldsymbol{{\beta}}^{t})\bigr)$ is equivalent to $ \diag\bigl(\mathbf{C}_{t}\bigr)\odot \boldsymbol{{\beta}}^{t} $ with complexity  $O(M^{2}+N^{2})$, which is actually the complexity of $ \nabla_{\thetv^{t}}\mathrm{SE} $ and $ \nabla_{\thetv^{t}}\mathrm{SE} $. In a similar way, we obtain  the complexity of $\nabla_{\thetv^{u}}I_{k}$, $u\in\{t,r\}$, $ \nabla_{\alv}S_{k} $, and $ \nabla_{\alv}I_{k} $. Thus, the overall complexity is $O(K(N^2+M^2))$.

\subsection{Optimization of ASTARS for MS protocol}
For the MS protocol, the amplitude values  have to be binary, which means that  $\betv^{t}_n\in \{ 0, 1\} $ and $\betv^{r}_n\in \{ 0, 1\} $. This problem corresponds to binary nonconvex programming that  is NP-hard, and very difficult to solve. For this kind of problem, we resort to practical solutions. Herein,  we will  round off the solution from \eqref{Maximization} to the closest binary value. As can be seen, in the next section, this solution results in a quite good performance. A topic for future work could be the solution to this problem by more advanced methods.

\section{Numerical Results}\label{Numerical} 
This section presents an illustration and discussion of  the analytical results corresponding to the downlink performance in terms of the achievable sum SE of an ASTARS-aided mMIMO system with imperfect CSI and correlated Rayleigh fading. In parallel, in certain figures, we  provide Monte-Carlo (MC) simulations ($ 10^{3} $ independent channel realizations) represented by "\ding{53}" marks. The simulations corroborate our results and the tightness of the UatF bound. We provide the following scenarios for comparison.
\begin{itemize}
	\item ASTARS-Random phase shifts: We assume that the phase shifts  of the elements of the  ASTARS are randomly set.
	\item ASTARS-Random amplitudes: We assume that the amplitudes of the ASTARS are randomly set.
	\item Active-Random AAC: We assume that the AACs of the elements of the ASTARS are randomly set.
	\item Passive STAR-RIS \cite{Papazafeiropoulos2023a}: We adopt the algorithm in \cite{Papazafeiropoulos2023a}, which corresponds to a passive RIS-aided MU-MISO
	system with $\al_{n}=1, \forall n  $.
	\item No STAR-RIS: This scenario corresponds to the absence of RIS to present the benefit of the introduction of the STAR-RIS in the architecture.
\end{itemize}

\subsection{Simulation Setup}
Unless otherwise stated, we consider $ M=64 $, $ N=400 $, $ K=12 $, $ \tau_{\mathrm{c}}=196$, $ \tau=K $. Also, the transmit power of the pilot and data signals per UE is $ P=20~\mathrm{dBm} $. $ P_R $ is set as $ 50\% $ of $ P $, while $ P_{n} =2\mathrm{mW}, \forall n \in N$. The noise powers are $ \sigma_{v}^{2}=-160~\mathrm{dBm}  $ and $ \sigma^{2}=-174~\mathrm{dBm} $. \textcolor{black}{We consider a system operating
at a carrier frequency of $ 6~\mathrm{GHz} $, which means that $ \lambda_{BS}=0.2~\mathrm{m} $. } The setup assumes that the BS is located at the origin, i.e., at $ (0,0) $, while the ASTARS is located at $ (0, 700~\mathrm{m}) $. Regarding the locations of the $ K $ UEs, these are randomly distributed in a circle centred at the ASTARS of radius $ 10~\mathrm{m} $.\footnote{\textcolor{black}{The radiating near field region is given by \cite{Bjoernson2021}
\begin{align}
	0.62\sqrt{\frac{(a^{2}+b^{2})^{3/2}}{\lambda}}<d<\frac{2(a^{2}+b^{2})}{\lambda},
\end{align}
where $ a $, $ b $ denote the dimensions of the ASTARS. Since this inequality gives $ 0.29~\mathrm{m}<d<5~\mathrm{m} $ if we assume $ 20 $ elements per dimension, the UEs are found at the far-field region.}} The path-loss factors are obtained as $ \tilde{\beta}_{g}=10^{-3}d_{g}^{-2.5} $, $ \tilde{\beta}_{k}=10^{-3}\tilde{d}_{k}^{-2} $, and $ \bar{\beta}_{k}=10^{-3}\bar{d}_{k}^{-4} $, where $ d_{g} $, $ \tilde{d}_{k} $, and $ \bar{d}_{k} $ describe the distances between the BS and the ASTARS, the ASTARS and UE $ k $, and the BS and UE $ k $, respectively. The AoA and AoD parameters  for the LoS channel are generated randomly from $ [0, 2\pi). $  The correlation matrices $ \bR_{\mathrm{BS}}$ and $\bR_{\mathrm{RIS}} $ are  computed according to \cite{Hoydis2013} and \cite{Bjoernson2020}, respectively. Note that  the dimensions of each ASTARS element are $ d_{\mathrm{H}}\!=\!d_{\mathrm{V}}\!=\!\lambda/4 $.

\begin{figure}%
	\centering
	\subfigure[]{	\includegraphics[width=0.9\linewidth]{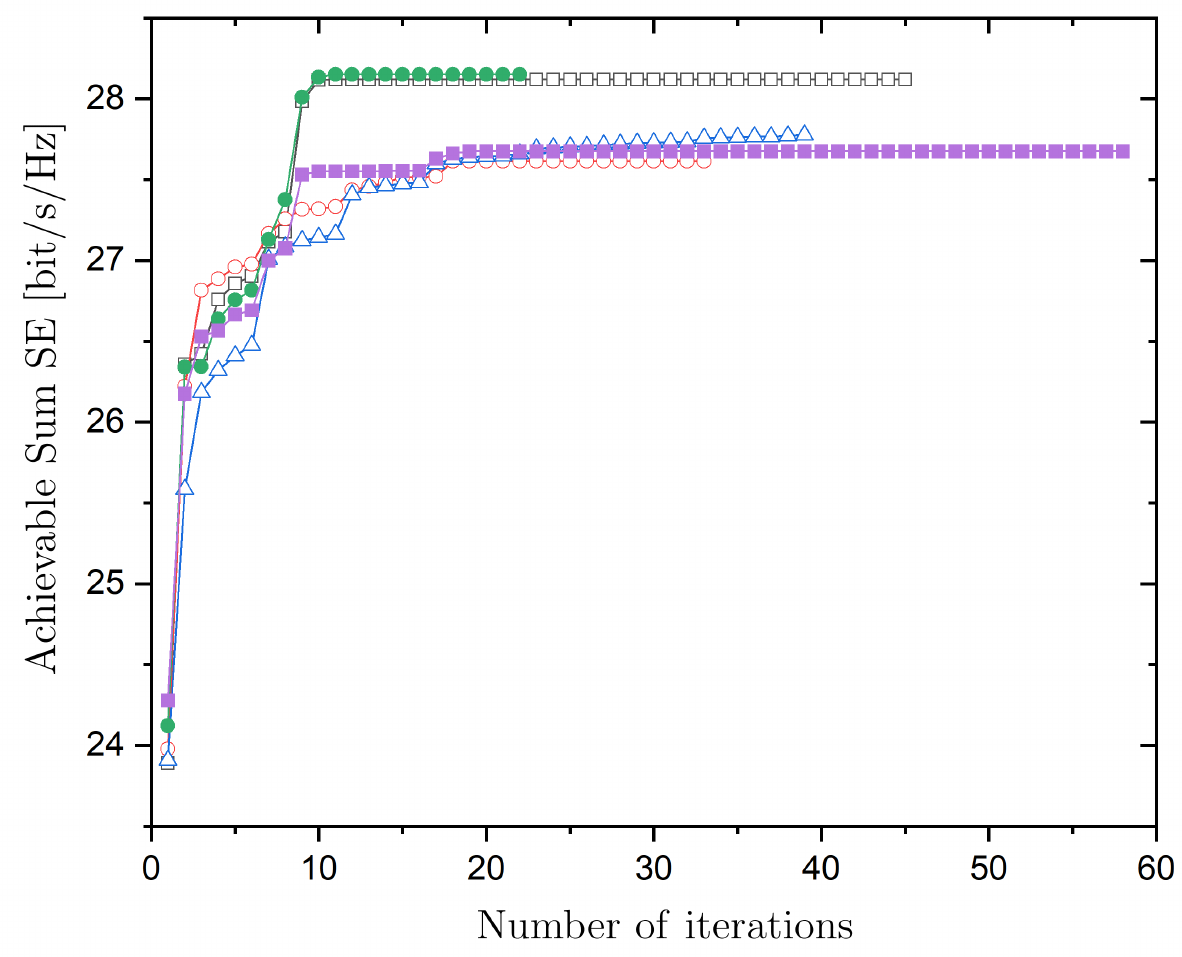}}\qquad
	\subfigure[]{	\includegraphics[width=0.9\linewidth]{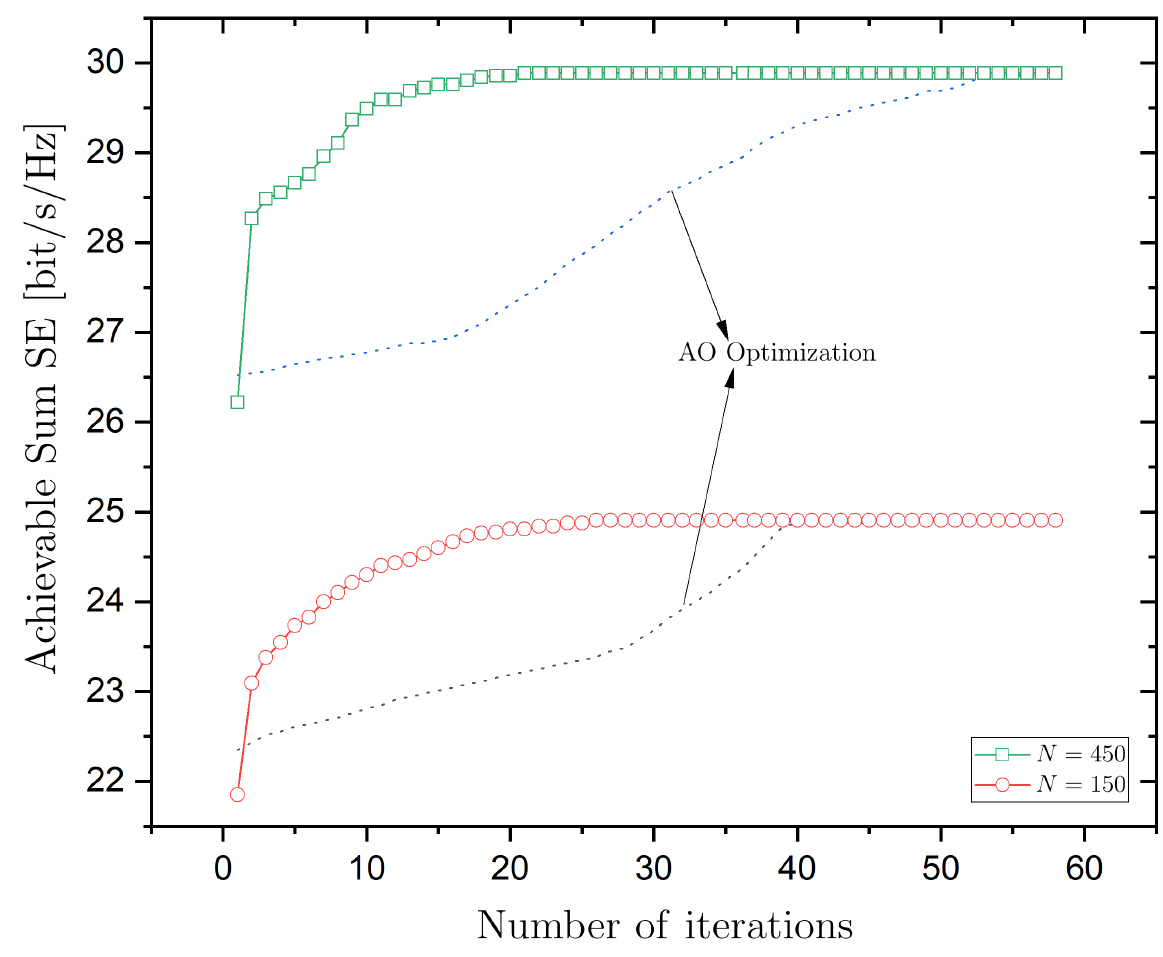}}\\
		\subfigure[]{	\includegraphics[width=0.9\linewidth]{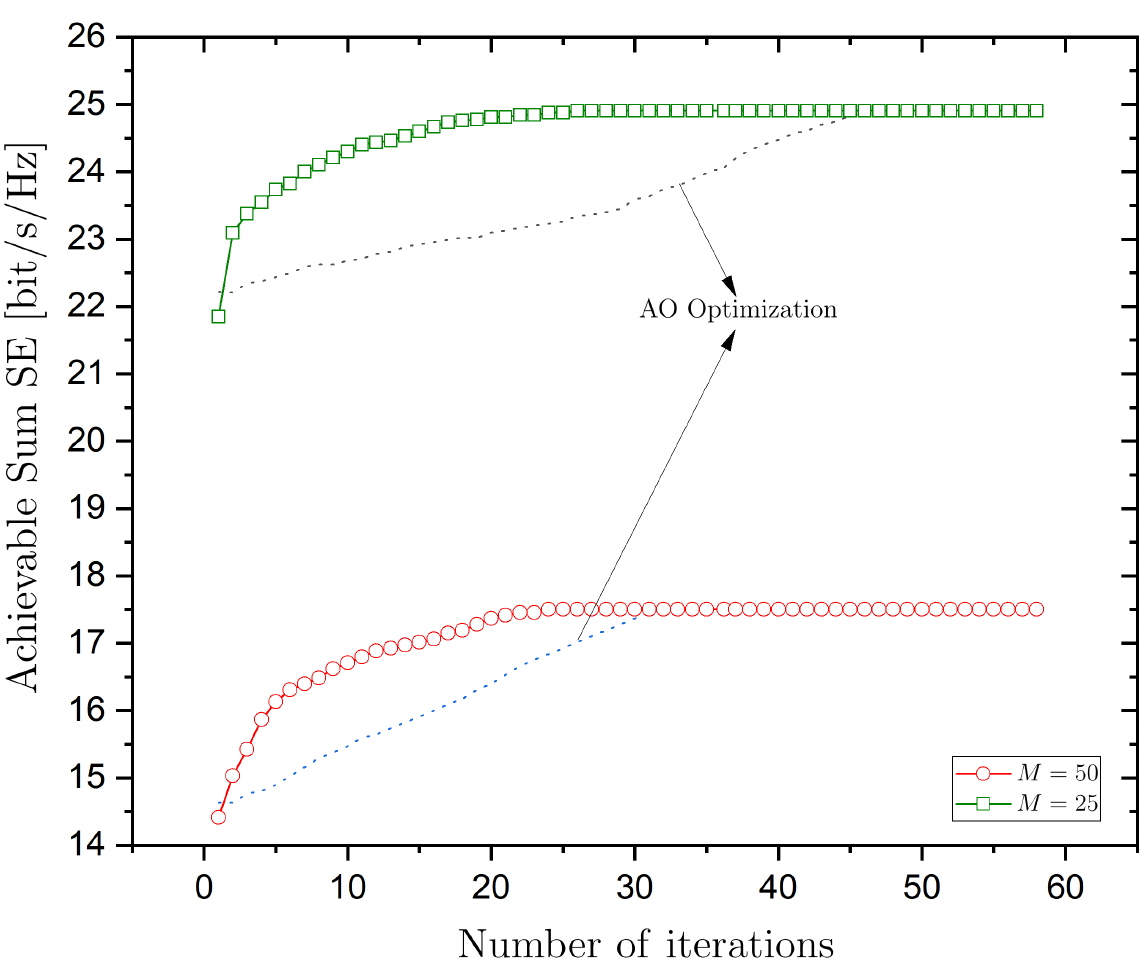}}\qquad
	\caption{Convergence of Algorithm \ref{Algoa1}: (a) Sensitivity of initial points, (b) Comparison with AO for different $ N $, (c) Comparison with AO for different $ M $. }
	\label{fig5}
\end{figure}

In Fig. \ref{fig5}.(a), we investigate  the convergence of the PGAM. In particular,  we depict the achievable  sum SE versus the iteration count  from $ 5 $ different initial points that are randomly generated. Regarding the initial points, we generate them  as described below.  For  the amplitudes, we  assume equal power splitting between transmission and reception mode for all elements of the STAR-RIS, i.e., $[\betv^{(0)}_r]_n=[\betv^{(0)}_t]_n=\sqrt{0.5} \forall n\in\mathcal{N}$. In the case of the  phase shifts, we set $[\thetv^{(0)}_{r}]_n=e^{j\phi_n^r}$ and  $[\thetv^{(0)}_{t}]_n=e^{j\phi_n^t}$. The phases $\phi_n^r$ and $\phi_n^t$ are independently distributed according to  the uniform distribution over $[0,2\pi]$. The algorithm terminates if the number of iterations is larger than $200$ or if the  gap  between the two last iterations is less than $10^{-5}$. Given that Problem $ (\mathcal{P}) $ is   nonconvex, we can lead to a stationary solution, which   may not be optimal. This means that if Algorithm \ref{Algoa1} starts  from different initial points, it will converge to different points as can be seen in Fig.~\ref{fig5}. Mitigation of this performance sensitivity  with respect to  the initial points of Algorithm  \ref{Algoa1} can be achieved by running  the algorithm from different starting points and selecting the best  solutions that converge. Overall, to obtain a good trade-off between complexity and SE  , extensive simulations showed  that  it is preferable to run    Algorithm \ref{Algoa1}   from $ 5 $  starting points that are randomly generated. \textcolor{black}{In Fig. \ref{fig5}.(b),
it is shown that more iterations are required by increasing the number of elements  due to the increasing complexity of  Algorithm \ref{Algoa1} and increasing size of the search space. Notably,  it is shown that the gradient-based optimization method converges fast to the optimum SE, while the AO method needs more iterations. Furthermore, initially, the AO returns a higher value because the starting points for the AO method have been chosen to result in a higher SE, while for the PGAM, the initial rate occurs by the initial values of its parameters. Also, as the number of elements $ N $ increases, the AO method tends to the optimum SE slower. The superiority of the proposed algorithm can be also shown in Fig. \ref{fig5}.(c), where we vary the number of BS antennas. }

\begin{figure}[!h]
	\begin{center}
		\includegraphics[width=0.9\linewidth]{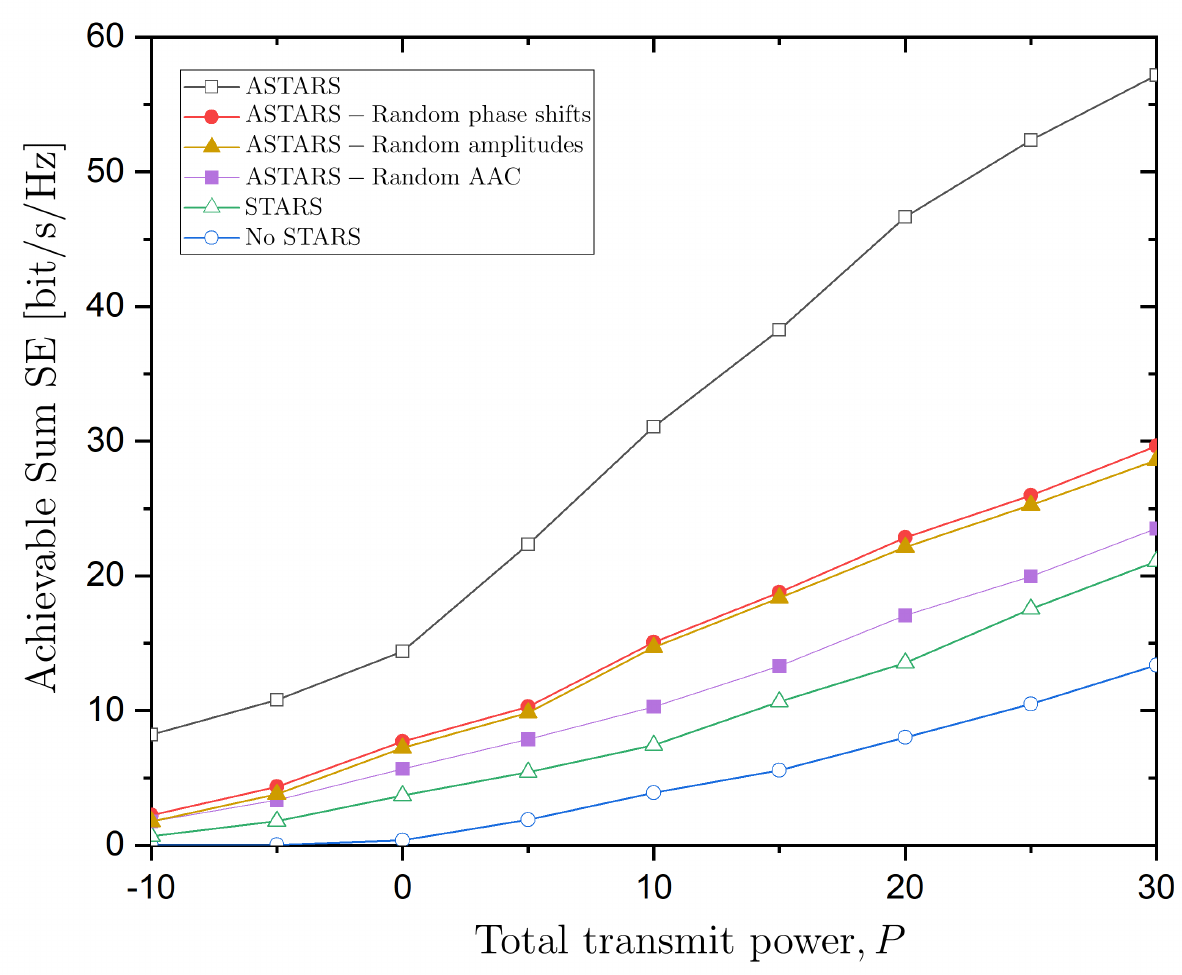}
		\caption{{Achievable sum SE versus the SNR. }}
		\label{fig1}
	\end{center} 
\end{figure}

In  Fig. \eqref{fig1}, we assess the achievable sum SE in ASTARS-assisted versus the total transmit power $ P $ for different scenarios. Clearly, the  sum SE increases with power $ P $ in all scenarios, but the case with optimized parameters corresponds to  significant performance gains. Concerning the random phase shifts, we observe a substantial loss in performance, while a similar loss is noticed  in the case of random amplitudes. The worst performance is shown if the AACs are not optimized in the ASTARS but are randomly set. In the same figure, we have drawn the case of STAR-RIS, i.e., when the surface has no active elements, and the case, where no surface is deployed. Given that the direct path between the BS and the UEs is rather weak in such cases, the SE is low. 

\begin{figure}%
	\centering
	\subfigure[]{	\includegraphics[width=0.9\linewidth]{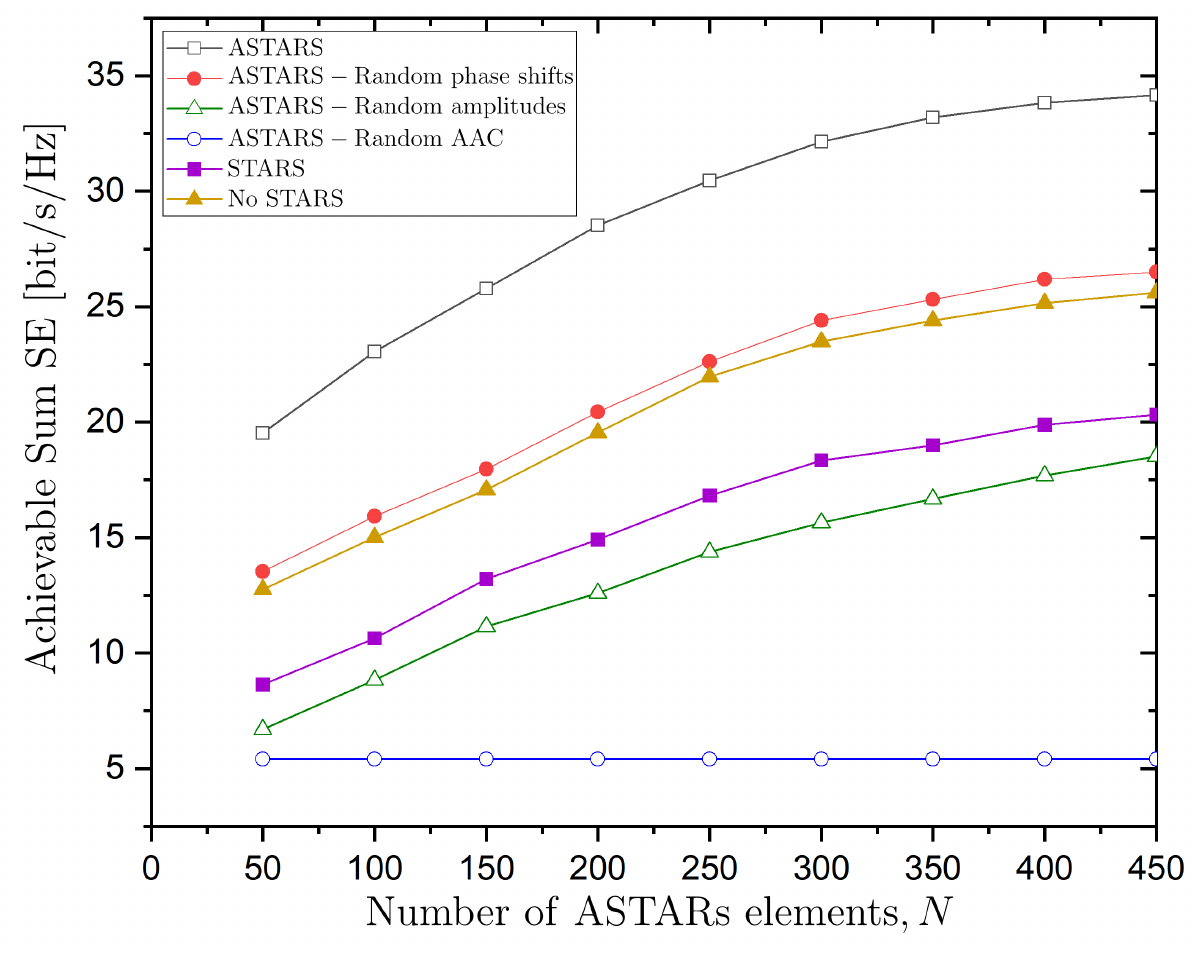}}\qquad
	\subfigure[]{	\includegraphics[width=0.9\linewidth]{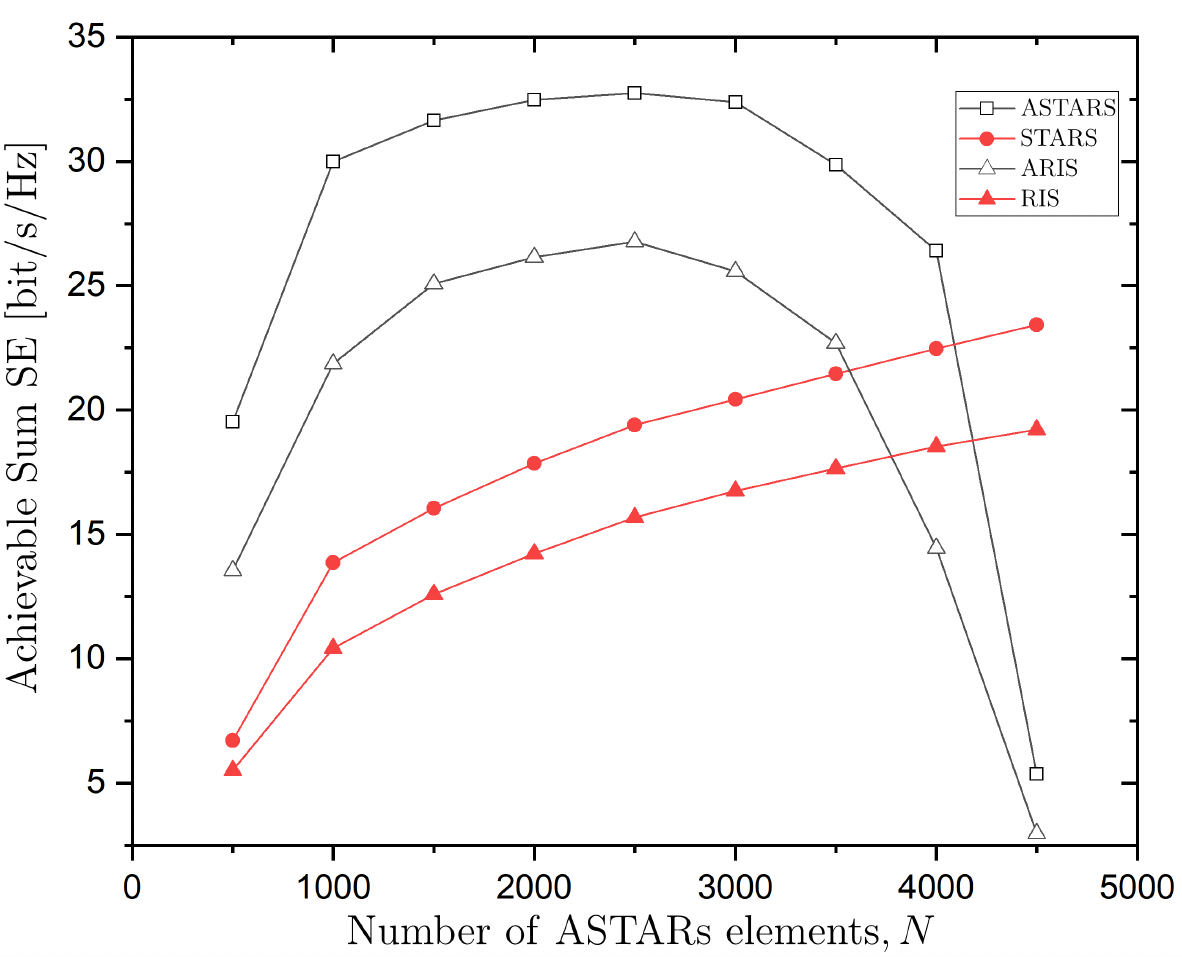}}\\
	\caption{Achievable sum SE versus the number of elements $ N  $ for: (a) conventional $ N $, (b) large $ N $ (Analytical results). }
	\label{fig2}
\end{figure}

Given the setup used in the previous figure, in Fig. \eqref{fig2}.(a), we illustrate the achievable sum SE versus the number of ASTARS elements $ N $. We notice that the performances for both ASTARS and passive STAR-RIS increase with  $ N $, but the performance of ASTARS is much better than that for passive STAR-RIS, i.e., the ASTAR is superior when $ N  $ is not very large. The reason is that, in this regime, the ASTARS  benefits from a small amount of power and achieves to  improve the received signal strength. Again, we observe that the ACCs have the largest impact compared to the other two parameters of the ASTARS, which are the phase shifts and the amplitudes of its elements. In Fig. \eqref{fig2}.(b), we illustrate the SE for large $ N $.  In this regime, the SE  of the ASTARS starts to decrease while passive STAR-RIS is superior because a large amount of power is consumed by the ASTARS to feed with power its amplifiers and because of the larger active thermal noise term. Overall, we observe that ASTARS is more advantageous than passive STAR-RIS in the case of the conventional number of surface elements. \textcolor{black}{For the sake of comparison, in the same figure, we have the depicted the cases of conventional RIS and active RIS (ARIS). Clearly, the combinations of active and STAR exhibits the best performance. Note that the conventional RIS consists
	of transmitting-only or reflecting-only elements, each with
	$  N_{t} $ and $  N_{r} $ elements, such that $  N_{t}+ N_{r}=N $. }

\begin{figure}[!h]
	\begin{center}
		\includegraphics[width=0.9\linewidth]{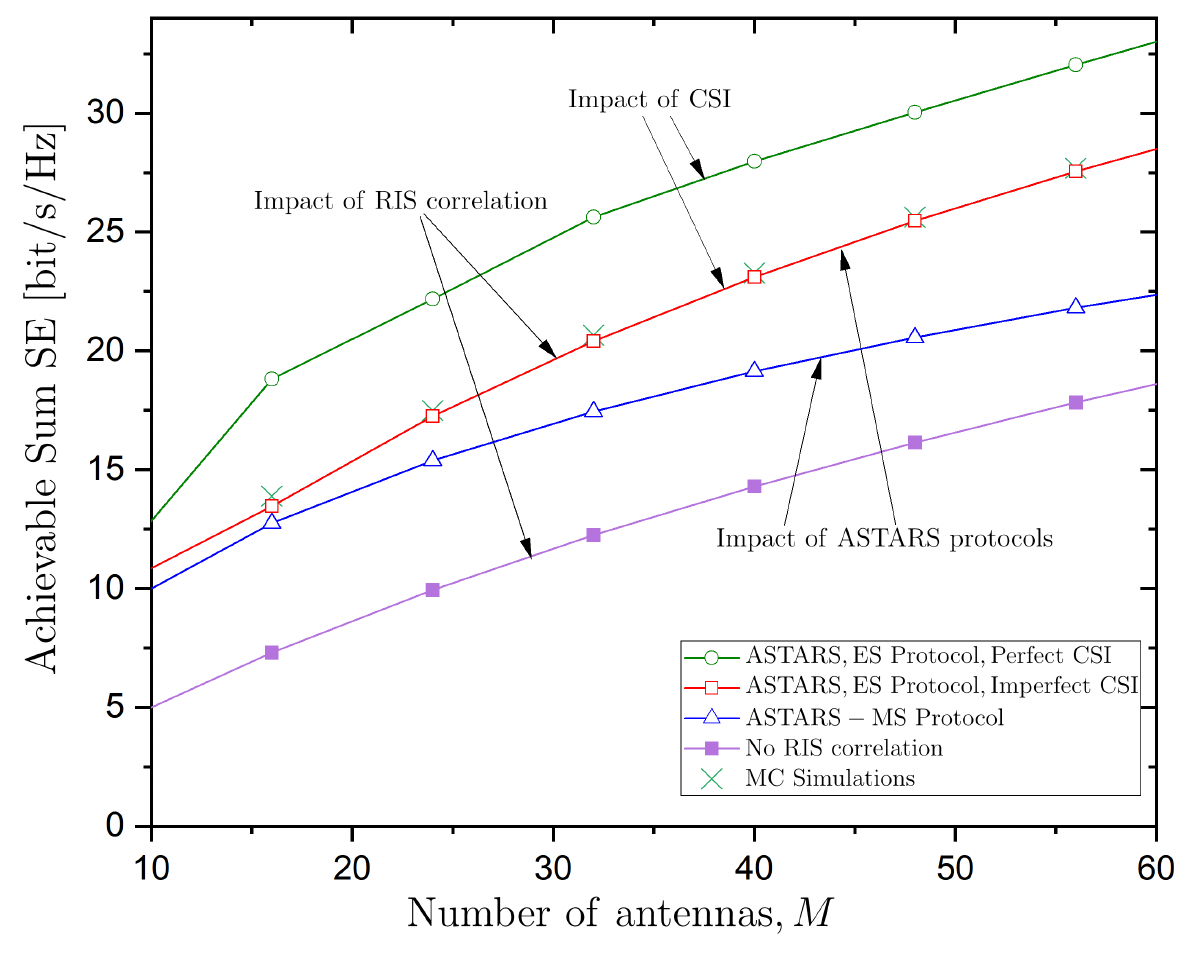}
		\caption{{Achievable sum SE versus the versus the number of BS antennas $M$. }}
		\label{fig4}
	\end{center} 
\end{figure}

Fig. \ref{fig4}  depicts the achievable sum SE  versus the number of BS antennas $ M $ while studying certain  conditions, which are the   ASTARS correlation, the ASTARS different protocols, and the CSI quality. In all cases, the SE increases with $ M $. If we assume no ASTARS correlation, there is no capability for phase shift optimization. Hence, the  performance is quite low. Regarding the ASTARS protocols, we observe that the  ES protocol results in  better performance compared to the MS protocol.  However, better performance comes with higher complexity. Moreover, concerning the impact of CSI, we have drawn the case of perfect CSI, which provides higher performance as expected.

\section{Conclusion}\label{Conclusion}
In this work, we introduced a new architecture enjoying the benefits of active RIS and STAR-RIS, denoted as ASTARS. Specifically, it can provide full-space coverage while mitigating the double fading effect. We provided a tailored channel estimation based on the MMSE approach, and we obtained a closed form for the achievable sum SE. Notably, we presented an iterative optimization algorithm  with low complexity, where we optimized all parameters simultaneously. The proposed method requires a low number of iterations to the usually AO method. Moreover,  we elaborated on useful insights with respect to the impact of various fundamental parameters on the ASTARS, and we showed its superiority compared to the conventional STAR-RIS in practical cases.
\appendices
\section{Proof of Lemma~\ref{PropositionDirectChannel}}\label{lem1}
Since the aggregated channel $ {\bh}_{k} $ is a Gaussian distributed random variable and the received signal in \eqref{train2} obeys the complex Bayesian linear model \cite[Eq. (15.63)]{Kay}, we  can  directly apply the results in \cite{Kay} to derive the MMSE channel
estimate of  $ {\bh}_{k} $.
Minimization of 
the MSE provides the MMSE estimate of $ \bh_{k} $ as in \cite[Eq. 15.64]{Kay}, which gives
\begin{align}
\hat{\bh}_{k} =\EE\!\left\{\br_{k}\bh_{k}^{\H}\right\}\left(\EE\!\left\{\br_{k}\br_{k}^{\H}\right\}\right)^{-1}\br_{k}.\label{Cor6}
\end{align}
Also, we have
\begin{align}
\EE\left\{\br_{k}\bh_{k}^{\H}\right\}
&=\EE\left\{\bh_{k}\bh_{k}^{\H}\right\}
\nn\\&=\bR_{k}\label{Cor0}
\end{align}
because they are uncorrelated.
The remaining term in \eqref{Cor6} is obtained as
\begin{align}
\bQ_{k}^{-1}&\triangleq\EE\left\{\br_{k}\br_{k}^{\H}\right\}\nn\\
& =\bR_{k} +{\tilde{ \beta}_{g}}\sigma_{v}^{2}\sum_{j=1}^{K}\tr(\bR_{\mathrm{RIS}}\bA^2\bC_{w_{j}})\bR_{\mathrm{BS}} + \frac{\sigma^2}{ \tau p }\Id_{M},\label{Cor1}
\end{align}
where $ \bC_{w_{k}}=\diag((\beta_{1}^{w_{j}})^{2}, \ldots,  (\beta_{N}^{w_{j}})^{2}) $. 
Substituting   \eqref{Cor0} and \eqref{Cor1} into \eqref{Cor6}, we obtain \eqref{estim1}.  Regarding the  estimated channel, we obtain its covariance matrix as
\begin{align}
\bPsi_{k}\triangleq\EE\left\{\hat{\bh}_{k}	\hat{\bh}_{k}^{\H}\right\}=\bR_{k}\bQ_{k}\bR_{k}.\label{var1}
\end{align} 
The MMSE matrix is obtained as
\begin{align}
\mathrm{MSE}_{k}=\bR_{k}-\bR_{k}\bQ_{k}\bR_{k}
\end{align}
Also, given that the channel $ {\bh}_{k} $ is a Gaussian vector, the channel estimate $ \hat{\bh}_{k} $ and the estimation error $ \tilde{\bh}_{k} $ are independent of each other because of  the orthogonality principle of the MMSE estimator, which completes the proof.

\section{Proof of Theorem~\ref{TheoremRate}}\label{th1}
We start with  the derivation of $ \lambda $. It is obtained as
\begin{align}
\!\!\!	\lambda&=\frac{1}{\sum_{i=1}^{K}\!\EE\{\mathbf{f}_{i}^{\H}\mathbf{f}_{i}\}}\nn\\
&=\frac{1}{\sum_{i=1}^{K}\!\mathbb{E}\{\hat{\mathbf{h}}_{i}^{\H}\hat{\mathbf{h}}_{i}\}}\nn\\
&=\frac{1}{\sum_{i=1}^{K}\!\tr(\boldsymbol{\Psi}_{i})}.\label{normalization}
\end{align}

Regarding the numerator of \eqref{gammaSINR}, we  can  write $ {{S}}_k $ as
\begin{align}
{{S}}_{k}&=|\EE\{\bh_{k}^{\H}\hat{\bh}_{k}\}|^{2}\nn\\
&=|\tr\big( \EE\{\hat{\bh}_{k} {\bh}_{k}^{\H} \} \!\big)|^{2} \\
&=|\tr\left( \EE\left\{\bR_{k}\bQ_{k} \br_{k}{\bh}_{k}^{\H}\right\} \right)\!|^{2}\label{term1}\\
&=|\tr\left(\bPsi_{k}\right)\!|^{2}\label{term2},
\end{align}
where, in \eqref{term1}, we have inserted \eqref{estim1}. In \eqref{term2}, we have computed the expectation between $ \br_{k} $ and $ \bh_{k} $.

Regarding $ I_k $ in \eqref{Int}, the first term of  can be derived as
\begin{align}
&\!\!\EE\big\{ \big| {\bh}_{k}^{\H}\hat{\bh}_{k}-\EE\big\{
{\bh}_{k}^{\H}\hat{\bh}_{k}\big\}\big|^{2}\big\}\!=\!
\EE\big\{ \big| {\bh}_{k}^{\H}\hat{\bh}_{k}\big|^{2}\big\}\!-\!\big|\EE\big\{
{\bh}_{k}^{\H}\hat{\bh}_{k}\big\}\big|^{2} \label{est2}\\
&=\EE\big\{ \big| \hat{\bh}_{k}^{\H} \hat{\bh}_{k} +\tilde{\bh}_{k}^{\H}\hat{\bh}_{k}\big|^{2}\big\}-\big|\EE\big\{
\hat{\bh}_{k}^{\H}\hat{\bh}_{k}\big\}\big|^{2}\label{est3} \\
&=\EE\big\{|\hat{\bh}_{k}^{\H}\hat{\bh}_{k}|^{2}\big\}+\EE\big\{|\tilde{\bh}_{k}^{\H}\hat{\bh}_{k}|^{2}\big\}-\big|\EE\big\{
\hat{\bh}_{k}^{\H}\hat{\bh}_{k}\big\}\big|^
{2} \label{est5}\\
&=\tr\!\left( \bR_{k}\bPsi_{k}\right).\label{est4}
\end{align}
In \eqref{est2},  we have applied the property $\EE\left\{ |X+Y|^{2}\right\} =\EE\left\{ |X|^{2}\right\} +\EE\left\{ |Y^{2}|\right\}$, which holds between two independent random variables when one of them has zero mean value, e.g., $\EE\left\{ X\right\}=0 $. In \eqref{est4}, we have applied \cite[Lem. 1]{Bjoernson2014}, which gives $ \EE\big\{|\hat{\bh}_{k}^{\H}\hat{\bh}_{k}|^{2}\big\}=\tr(\bPsi_{k}^{2})+\tr^{2}(\bPsi_{k}) $. Also, the last equation results by taking advantage of the independence between the two random vectors and since $ \EE\big\{|\hat{\bh}_{k}^{\H}\tilde{\bh}_{k}|^{2}\big\}=\tr\left(\bPsi_{k}\left(\bR_{k}-\bPsi_{k}\right)\right) $.

The second term of  \eqref{Int} can be obtained as
\begin{align}
\EE\big\{ \big| {\bh}_{k}^{\H}\hat{\bh}_{i}\big|^{2}\big\}
&=		\EE\big\{ \big| \bd_{k}^{\H}\hat{\bh}_{i}\big|^{2}\big\}+	\EE\big\{ \big|\bq_{k}^{\H} \bA\bPhi_{w_{k}}^{\H} \bG^{\H}\hat{\bh}_{i}\big|^{2}\big\}\label{MUI},
\end{align}
where the first term is obtained as
\begin{align}
\EE\big\{ \big| \bd_{k}^{\H}\hat{\bh}_{i}\big|^{2}\big\}=\bar{ \beta}_{k}\tr(\bR_{\mathrm{BS}}\bPsi_{i})\label{first}.
\end{align}
In \eqref{first}, we have considered the independence between $\bd_{k} $ and $ \hat{\bh}_{i} $.

The second term of \eqref{MUI} becomes
\begin{align}
&	\EE\big\{ \big|\bq_{k}^{\H} \bA\bPhi_{w_{k}}^{\H} \bG^{\H}\hat{\bh}_{i}\big|^{2}\big\}=	\EE\big\{ \bq_{k}^{\H} \bA\bPhi_{w_{k}}^{\H} \bG^{\H}\bPsi_{i}\bG\bPhi_{w_{k}}^{\H}\bA\bq_{k}\big\}\label{second1}\\
&=	\hat{\beta}_{k}\tr(\bR_{\mathrm{BS}}\bPsi_{i}) \ba_{N}^{\H} \bA\bPhi_{w_{k}}^{\H} \bR_{\mathrm{RIS}}\bPhi_{w_{k}}^{\H}\bA\ba_{N}\label{second2},
\end{align}
where, in \eqref{second1}, we have computed first the expectation with respect to $ \hat{\bh}_{i} \hat{\bh}_{i} ^{\H}$ since $ \hat{\bh}_{i} $ and $ \bG $ are independent. Next, to compute the expectation with respect to $ \bG $, we have applied the property $ \EE\{\bV \bU\bV^{\H}\} =\tr (\bU) \Id_{M}$. Hence, \eqref{MUI} is written as
\begin{align}
\EE\big\{ \big| {\bh}_{k}^{\H}\hat{\bh}_{i}\big|^{2}\big\}
&=	\bar{ \beta}_{k}\tr(\bR_{\mathrm{BS}}\bPsi_{i})+\hat{\beta}_{k}\tr(\bR_{\mathrm{BS}}\bPsi_{i}) \nn\\
&	\times \tr(\ba_{N}^{\H} \bA\bPhi_{w_{k}}^{\H} \bR_{\mathrm{RIS}}\bPhi_{w_{k}}^{\H}\bA\ba_{N})\\
&=\tr(\bR_{k}\bPsi_{i}),\label{second3}
\end{align}
where the last equation is obtained due to \eqref{cov1}.

The third term in the denominator can be written as
\begin{align}
\EE\{\big\|\bq_{k}^{\H}\bPhi_{w_{k}}^{\H}\bA\big\|_{2}^{2}\}&=\tilde{ \beta}_{k}\tr(\ba_{N}^{\H}\bA^{2}\bC_{w_{k}}\ba_{N})\\
&=\tilde{ \beta}_{k}\sum_{i=1}^{N}\al_{i}(\beta_{i}^{w_{k}})^{2}\label{norm},
\end{align}
where $ \bC_{w_{k}}=\diag( {\beta_{1}^{w_{k}}} , \ldots,  {\beta_{N}^{w_k}}) $. Also, we have the following lemma.

Substitution of \eqref{est4}, \eqref{second3}, and \eqref{norm} into \eqref{Int}  concludes the proof.

\section{Proof of Lemma~\ref{LemmaGradients}}\label{lem2}
Firstly, we focus on the derivation of $\nabla_{\boldsymbol{\theta^{t}}}\mathrm{SE} $, which is the   gradient of the achievable sum SE regarding $ \boldsymbol{\theta}^{t\ast}$. Easily, we obtain
\begin{equation}
\nabla_{\thetv^{t}}\mathrm{SE}=c\sum_{k=1}^{K}\frac{{I_{k}}\nabla_{\boldsymbol{\theta}^{t}}S_{k}-S_{k}\nabla_{\boldsymbol{\theta}^{t}}{I_{k}}}{(1+\gamma_{k}){I_{k}}^{2}},
\end{equation}
where $c=\frac{\tau_{c}-\tau}{\tau_{c}\log_{2}(e)}$ and we have denoted that $ \gamma_{k}$ is the downlink  signal-to-interference-plus-noise ratio (SINR) given from Theorem \ref{TheoremRate}.

The 	 calculation of $\nabla_{\thetv^{t}}S_k$ requires the derivation of the   complex differential of $ S_k $. We have  
\begin{align}
d(S_{k})&=d\bigl(\tr(\boldsymbol{\Psi}_{k})^{2}\bigr)=2\tr(\boldsymbol{\Psi}_{k})d\tr(\boldsymbol{\Psi}_{k})\nn\\
&=2\tr(\boldsymbol{\Psi}_{k})\tr(d\boldsymbol{\Psi}_{k}).\label{eq:dSk}
\end{align}

For the derivation of $ d(\bPsi_{k}) $, application of \cite[Eq. (3.35)]{hjorungnes:2011} yields
\begin{align}
\!\!\!d(\boldsymbol{\Psi}_{k})
&=d(\mathbf{R}_{k})\mathbf{Q}_{k}\mathbf{R}_{k}+\mathbf{R}_{k}d(\mathbf{Q}_{k})\mathbf{R}_{k}+\mathbf{R}_{k}\mathbf{Q}_{k}d(\mathbf{R}_{k}).\label{eq:dPsik}
\end{align}
The differentials in \eqref{eq:dPsik} are obtained below. Since  $w_k=t$, we can find \eqref{cov1} that
\begin{align}
\bR_{k}
&=\bar{ \beta}_{k}\bR_{\mathrm{BS}}+\hat{\beta}_{k}\tr(\bC_{t}\bPhi_{t}^{\H}  )\bR_{\mathrm{BS}}\label{cov1t}.
\end{align}
where $\hat{\beta}_{k}=\tilde{\beta}_{g}\tilde{\beta}_{k}$ and $\mathbf{C}_{t}=\bA\mathbf{R}_{\mathrm{RIS}}\bA\bPhi_{t} \ba_{N}\ba_{N}^{\H}$. Similarly, if   $w_k=r$, we define $\mathbf{C}_{r}=\bA\mathbf{R}_{\mathrm{RIS}}\bA\bPhi_{r} \ba_{N}\ba_{N}^{\H}$. The differential of $ \bR_{k} $ is written as
\begin{align}
d(\mathbf{R}_{k})\label{dRk:general}  
& =\hat{\beta}_{k}\mathbf{R}_{\mathrm{BS}}\tr\bigl(\mathbf{C}_{t}\herm d(\bPhi_{t})+\mathbf{C}_{t}d\bigl(\bPhi_{t}\herm\bigr)\bigr)\\
&=\hat{\beta}_{k}\mathbf{R}_{\mathrm{BS}}\bigl(\bigl(\diag\bigl(\mathbf{C}_{t}\herm\diag(\boldsymbol{{\beta}}^{t})\bigr)\bigr)\trans d(\boldsymbol{\theta}^{t})\nn\\
&+\bigl(\diag\bigl(\mathbf{C}_{t}\diag(\boldsymbol{{\beta}}^{t})\bigr)\bigr)\trans d(\boldsymbol{\theta}^{t\ast})\bigr),\label{eq:dRk}
\end{align}
where \eqref{eq:dRk} is obtained because  $\bPhi_{t}$ is diagonal. 

Application of  \cite[eqn. (3.40)]{hjorungnes:2011} gives
\begin{align}
d(\mathbf{Q}_{k})  &=d\left(\!\!\bigl(\mathbf{R}_{k}+\frac{\sigma}{\tau p}\mathbf{I}_{M}\bigr)\!\!\right)^{-1}\nn\\
&=-\bigl(\mathbf{R}_{k}+\frac{\sigma^{2}}{\tau p}\mathbf{I}_{M}\bigr)^{-1}d\bigl(\mathbf{R}_{k}+\frac{\sigma^{2}}{\tau p}\mathbf{I}_{M}\bigr)\bigl(\mathbf{R}_{k}+\frac{\sigma}{\tau p}\mathbf{I}_{M}\bigr)^{-1}\nonumber \\
& =-\mathbf{Q}_{k}d(\mathbf{R}_{k})\mathbf{Q}_{k}.\label{eq:dQk}
\end{align}
Combining (\ref{eq:dPsik}) and (\ref{eq:dQk})
provides
\begin{equation}
d(\boldsymbol{\Psi}_{k})=d(\mathbf{R}_{k})\mathbf{Q}_{k}\mathbf{R}_{k}\!-\!\mathbf{R}_{k}\mathbf{Q}_{k}d(\mathbf{R}_{k})\mathbf{Q}_{k}\mathbf{R}_{k}+\mathbf{R}_{k}\mathbf{Q}_{k}d\mathbf{R}_{k}.\label{eq:dPsik-1}
\end{equation}

Substitution of \eqref{eq:dPsik-1} and \eqref{eq:dRk} into  \eqref{eq:dSk} results in
\begin{align}
d(S_{k}) 
& =\nu_{k}\bigl(\diag\bigl(\bigl(\mathbf{C}_{t}\herm\diag(\boldsymbol{\mathbf{\beta}}^{t})\bigr)\bigr)\trans d\boldsymbol{\theta}^t\nn\\
&+\bigl(\diag\bigl(\mathbf{C}_{t}\diag(\boldsymbol{\mathbf{\beta}}^{t})\bigr)\bigr)\trans d\boldsymbol{\theta}^{t\ast}\bigr),\label{eq:dSk:final}
\end{align}
where 
\begin{equation}
\nu_{k}\!=\!2\hat{\beta}_{k}\tr(\boldsymbol{\Psi}_{k})\!\tr\bigl(\!\bigl(\mathbf{Q}_{k}\mathbf{R}_{k}+\mathbf{R}_{k}\mathbf{Q}_{k}-\mathbf{Q}_{k}\mathbf{R}_{k}^{2}\mathbf{Q}_{k}\bigr)\mathbf{R}_{\mathrm{BS}}\bigr).
\end{equation}

Thus, we obtain
\begin{align}
\nabla_{\thetv^{t}}S_k&=\frac{\partial}{\partial{\thetv^{t\ast}}}S_{k}\nn\\
&=\nu_{k}\diag\bigl(\mathbf{C}_{t}\diag(\boldsymbol{{\beta}}^{t})\bigr)
\end{align}
for $w_k=t$, which  proves \eqref{derivtheta_t}. The proof of \eqref{derivtheta_r} can be obtained by following the same procedure but it is omitted for the sake of brevity.

Regarding $\nabla_{\thetv^{t}}{I_{k}}$, the differential of $ I_{k} $ is written as
\begin{align}
d(I_{k})&=\sum\nolimits _{i=1}^{K}\tr(d(\mathbf{R}_{k})\boldsymbol{\Psi}_{i})+\sum\nolimits _{i=1}^{K}\tr(\mathbf{R}_{k}d(\boldsymbol{\Psi}_{i}))\nn\\
&+\Big(\sigma_{v}^{2}\tilde{ \beta}_{k}\sum_{i=1}^{N}\al_{i}(\beta_{i}^{w_{k}})^{2}+\sigma^{2}\Big)\frac{K}{ p}\sum_{i=1}^{K}\tr(d(\bPsi_{i})).
\label{Den2}
\end{align}

However, if $w_{i}\neq t$, we have $d(\boldsymbol{\Psi}_{i})=0$  because $\boldsymbol{\Psi}_{i}$ is independent of $\thetv^{t}$. Hence, \eqref{Den2} becomes
\begin{equation}
d({I_{k}})=\tr(\boldsymbol{\Psi}d(\mathbf{R}_{k}))+\sum\nolimits _{i\in\mathcal{K}_{t}}\tr(\bar{\mathbf{R}}_{k}d(\boldsymbol{\Psi}_{i})),	\label{ik1}
\end{equation}	where $\boldsymbol{\Psi}=\sum\nolimits _{i=1}^{K}\boldsymbol{\Psi}_{i}$
and $\bar{\mathbf{R}}_{k}=\mathbf{R}_{k}+\Big(\sigma_{v}^{2}\tilde{ \beta}_{k}\sum_{i=1}^{N}\al_{i}(\beta_{i}^{w_{k}})^{2}+\sigma^{2}\Big)\frac{K}{ p}\mathbf{I}_{M}$.

Inserting \eqref{eq:dPsik-1} into \eqref{ik1} gives
\begin{align}
d({I_{k}}) &=\tr(\boldsymbol{\Psi}d\mathbf{R}_{k})+\sum\nolimits _{i\in \mathcal{K}_t}\tr\big(\bar{\mathbf{R}}_{k}\bigl(d(\mathbf{R}_{i})\mathbf{Q}_{i}\mathbf{R}_{i}\nn\\
&-\mathbf{R}_{i}\mathbf{Q}_{i}d(\mathbf{R}_{i})\mathbf{Q}_{i}\mathbf{R}_{i}+\mathbf{R}_{i}\mathbf{Q}_{i}d(\mathbf{R}_{i})\bigr)\big)\nn\\
& =\tr({\boldsymbol{\Psi}}_{k}d(\mathbf{R}_{k}))+\sum\nolimits _{i\in \mathcal{K}_t}\tr\bigl(\tilde{\mathbf{R}}_{ki}d(\mathbf{R}_{i})\bigr),
\end{align}
where 
\begin{equation}
\tilde{\mathbf{R}}_{ki}=\mathbf{Q}_{i}\mathbf{R}_{i}\bar{\mathbf{R}}_{k}-\mathbf{Q}_{i}\mathbf{R}_{i}\bar{\mathbf{R}}_{k}\mathbf{R}_{i}\mathbf{Q}_{i}+\bar{\mathbf{R}}_{k}\mathbf{R}_{i}\mathbf{Q}_{i},i \in \mathcal{K}_t.
\end{equation}

Given that $d(\bR_{k})=0$ if $w_{k}\neq t$, we use (\ref{eq:dRk}),  and we can derive $\nabla_{\thetv^{t}}{I_{k}}$ as
\begin{align}
\nabla_{\thetv^{t}}{I_{k}}=\frac{\partial}{\partial\boldsymbol{\theta}^{t\ast}}{I_{k}} & =\diag\bigl(\tilde{\mathbf{A}}_{kt}\diag(\boldsymbol{\mathbf{\beta}}^{t})\bigr),
\end{align}
where $\bar{\nu}_{k}=\hat{\beta}_{k}\tr\bigl(\check{\boldsymbol{\Psi}}_{k}\mathbf{R}_{\mathrm{BS}}\bigr)$,
$\tilde{\nu}_{ki}=\hat{\beta}_{k}\tr\bigl(\tilde{\mathbf{R}}_{ki}\mathbf{R}_{\mathrm{BS}}\bigr)$,
and 
\begin{equation}
\tilde{\mathbf{A}}_{kt}=\begin{cases}
	\bar{\nu}_{k}\mathbf{A}_{t}+\sum\nolimits _{i\in\mathcal{K}_{t}}^{K}\tilde{\nu}_{ki}\mathbf{C}_{t} & w_{k}=t\\
	\sum\nolimits _{i\in\mathcal{K}_{t}}\tilde{\nu}_{ki}\mathbf{C}_{t} & w_{k}\neq t,
\end{cases}
\end{equation}
which proves \eqref{derivtheta_t_Ik} when $u=t$. Following the same procedure yields \eqref{derivtheta_r_Ik}, but again, we  omit the steps for brevity.

In the case of $\nabla_{\boldsymbol{\beta}^{t}}S_{k}$, we observe that $ d(\mathbf{R}_{k}) $ can be written as
\begin{align}
d(\mathbf{R}_{k}) & =\hat{\beta}_{k}\mathbf{R}_{\mathrm{BS}}\tr\bigl(\mathbf{C}_{t}\herm d(\bPhi_{t})+\mathbf{C}_{t}d\bigl(\bPhi_{t}\herm\bigr)\bigr)\\
& =\hat{\beta}_{k}\mathbf{R}_{\mathrm{BS}}\bigl(\diag\bigl(\mathbf{C}_{t}\herm\diag(\btheta^{t})\bigr)^{\T}d(\boldsymbol{\beta}^{t})\nn\\
&+\diag\bigl(\mathbf{C}_{t}\diag(\btheta^{t\ast})\bigr)^{\T}d(\boldsymbol{\beta}^{t})\bigr)\\
& =2\hat{\beta}_{k}\mathbf{R}_{\mathrm{BS}}\Re\bigl\{\diag\bigl(\mathbf{C}_{t}\herm\diag(\btheta^{t})\bigr\}^{\T} d(\boldsymbol{\beta}^{t}).\label{dRk_beta_t}
\end{align}
Hence, use of\eqref{dRk_beta_t} in \eqref{eq:dSk:final} results in 
\begin{equation}
\nabla_{\boldsymbol{\beta}^{t}}S_{k} = 2 \nu_k\Re\bigl\{\diag\bigl(\mathbf{C}_{t}\herm\diag(\btheta^{t})\bigr)\bigr\}.
\end{equation}
Note that  $\nabla_{\boldsymbol{\beta}^{r}}S_{k}$ can be obtained in a similar way.

For $\nabla_{\boldsymbol{\beta}^{t}}{I_{k}}$, after following the same procedure as above, we obtain
\begin{align}
\nabla_{\boldsymbol{\beta}^{t}}{I_{k}}&=2\Re\bigl\{\diag\bigl(\tilde{\mathbf{A}}_{kt}\herm\diag(\boldsymbol{\btheta}^{t})\bigr)\bigr\},
\end{align}
while $ 	\nabla_{\boldsymbol{\beta}^{r}}{I_{k}} $ is obtained similarly.

Analogously, we obtain  $\nabla_{\alv}S_{k}$. Specifically, we obtain
\begin{align}
d(\mathbf{R}_{k}) & =\hat{\beta}_{k}\mathbf{R}_{\mathrm{BS}}\tr\bigl(d(\bA)\bD+d(\bA)\bD^{\H}\bigr)\\
&=2\hat{\beta}_{k}\mathbf{R}_{\mathrm{BS}}\Re\bigl\{\diag\bigl(\bD\bigr)\}^{\T} d(\alv).
\end{align}
where $ \bD=\mathbf{R}_{\mathrm{RIS}}\bA\bPhi_{t} \ba_{N}\ba_{N}^{\H} \bPhi_{t}\herm $. Thus,  we obtain  $ \nabla_{\alv}S_{k} $ and $\nabla_{\alv}{I_{k}}$ as
\begin{align}
\nabla_{\alv}S_{k} &= 2 \nu_k\Re\bigl\{\diag\bigl(\bD\bigr)\},\\
\nabla_{\alv}{I_{k}}&=2\Re\bigl\{\diag\bigl(\tilde{\mathbf{A}}_{kt}\herm\diag(\bD)\bigr)\bigr\},\end{align}
which concludes the proof.
\bibliographystyle{IEEEtran}

\bibliography{IEEEabrv,mybib}

\begin{thebibliography}{10}
\providecommand{\url}[1]{#1}
\csname url@samestyle\endcsname
\providecommand{\newblock}{\relax}
\providecommand{\bibinfo}[2]{#2}
\providecommand{\BIBentrySTDinterwordspacing}{\spaceskip=0pt\relax}
\providecommand{\BIBentryALTinterwordstretchfactor}{4}
\providecommand{\BIBentryALTinterwordspacing}{\spaceskip=\fontdimen2\font plus
\BIBentryALTinterwordstretchfactor\fontdimen3\font minus
  \fontdimen4\font\relax}
\providecommand{\BIBforeignlanguage}[2]{{%
\expandafter\ifx\csname l@#1\endcsname\relax
\typeout{** WARNING: IEEEtran.bst: No hyphenation pattern has been}%
\typeout{** loaded for the language `#1'. Using the pattern for}%
\typeout{** the default language instead.}%
\else
\language=\csname l@#1\endcsname
\fi
#2}}
\providecommand{\BIBdecl}{\relax}
\BIBdecl

\bibitem{DiRenzo2020}
M.~Di~Renzo \emph{et~al.}, ``Smart radio environments empowered by
  reconfigurable intelligent surfaces: {How} it works, state of research, and
  the road ahead,'' \emph{IEEE J. Sel. Areas Commun.}, vol.~38, no.~11, pp.
  2450--2525, 2020.

\bibitem{Pan2021}
C.~Pan \emph{et~al.}, ``Reconfigurable intelligent surfaces for {6G} systems:
  Principles, applications, and research directions,'' vol.~59, no.~6, pp.
  14--20.

\bibitem{Wu2019}
Q.~Wu and R.~Zhang, ``Intelligent reflecting surface enhanced wireless network
  via joint active and passive beamforming,'' \emph{IEEE Trans. Wireless
  Commun.}, vol.~18, no.~11, pp. 5394--5409, 2019.

\bibitem{Huang2019}
C.~Huang \emph{et~al.}, ``Reconfigurable intelligent surfaces for energy
  efficiency in wireless communication,'' \emph{IEEE Transa. Wireless Commun.},
  vol.~18, no.~8, pp. 4157--4170, 2019.

\bibitem{Bjoernson2019b}
E.~Bj{\"o}rnson, {\"O}.~{\"O}zdogan, and E.~G. Larsson, ``Intelligent
  reflecting surface versus decode-and-forward: {How} large surfaces are needed
  to beat relaying?'' vol.~9, no.~2, pp. 244--248.

\bibitem{Han2019}
Y.~Han \emph{et~al.}, ``Large intelligent surface-assisted wireless
  communication exploiting statistical {CSI},'' vol.~68, no.~8, pp. 8238--8242.

\bibitem{Zhao2020}
M.-M. Zhao \emph{et~al.}, ``Intelligent reflecting surface enhanced wireless
  networks: {Two}-timescale beamforming optimization,'' \emph{IEEE Trans.
  Wireless Commun.}, vol.~20, no.~1, pp. 2--17, 2020.

\bibitem{Kammoun2020}
Q.~U.~A. {Nadeem} \emph{et~al.}, ``Asymptotic max-min {SINR} analysis of
  reconfigurable intelligent surface assisted {MISO} systems,'' \emph{IEEE
  Trans. Wireless Commun.}, vol.~19, no.~12, pp. 7748--7764, 2020.

\bibitem{Papazafeiropoulos2021}
A.~Papazafeiropoulos \emph{et~al.}, ``Intelligent reflecting surface-assisted
  {MU-MISO} systems with imperfect hardware: {Channel} estimation and
  beamforming design,'' \emph{IEEE Trans. Wireless Commun.}, vol.~21, no.~3,
  pp. 2077--2092, 2021.

\bibitem{Abrardo2021}
A.~Abrardo, D.~Dardari, and M.~Di~Renzo, ``Intelligent reflecting surfaces:
  {Sum}-rate optimization based on statistical position information,''
  \emph{IEEE Trans. Commun.}, vol.~69, no.~10, pp. 7121--7136, 2021.

\bibitem{Chen2021}
Y.~Chen \emph{et~al.}, ``{QoS}-driven spectrum sharing for reconfigurable
  intelligent surfaces ({RISs}) aided vehicular networks,'' \emph{IEEE Trans.
  Wireless Commun.}, vol.~20, no.~9, pp. 5969--5985, 2021.

\bibitem{Papazafeiropoulos2022a}
A.~Papazafeiropoulos \emph{et~al.}, ``Joint spatial division and multiplexing
  for {FDD} in intelligent reflecting surface-assisted massive {MIMO}
  systems,'' \emph{IEEE Trans. Veh. Tech.}, vol.~71, no.~10, pp.
  10\,754--10\,769, 2022.

\bibitem{Chen2022}
Y.~Chen \emph{et~al.}, ``Robust beamforming for active reconfigurable
  intelligent omni-surface in vehicular communications,'' \emph{IEEE J. Sel.
  Areas . Commun.}, vol.~40, no.~10, pp. 3086--3103, 2022.

\bibitem{Chen2022a}
------, ``Reconfigurable intelligent surface ({RIS})-aided vehicular networks:
  {Their} protocols, resource allocation, and performance,'' \emph{IEEE Veh.
  Tech. Mag.}, vol.~17, no.~2, pp. 26--36, 2022.

\bibitem{Wu2023}
H.~Wu \emph{et~al.}, ``Two-timescale beamforming optimization for downlink
  multi-user holographic {MIMO} surfaces,'' \emph{IEEE Trans. Veh. Tech.},
  2023.

\bibitem{Papazafeiropoulos2023}
A.~Papazafeiropoulos, I.~Krikidis, and P.~Kourtessis, ``Impact of channel aging
  on reconfigurable intelligent surface aided massive {MIMO} systems with
  statistical {CSI},'' \emph{IEEE Trans. Veh. Tech.}, vol.~72, no.~1, pp.
  689--703, 2023.

\bibitem{Long2021}
R.~Long \emph{et~al.}, ``Active reconfigurable intelligent surface-aided
  wireless communications,'' \emph{IEEE Trans. Wireless Commun.}, vol.~20,
  no.~8, pp. 4962--4975, 2021.

\bibitem{Zhang2022}
Z.~Zhang \emph{et~al.}, ``Active {RIS vs. passive RIS: Which will} prevail in
  {6G?}'' \emph{IEEE Trans. Commun.}, 2022.

\bibitem{You2022}
C.~You \emph{et~al.}, ``How to deploy intelligent reflecting surfaces in
  wireless network: {BS}-side, user-side, or both sides?'' \emph{J. Commun.
  Inf. Net.}, vol.~7, no.~1, pp. 1--10, 2022.

\bibitem{Bjoernson2019}
E.~Bj{\"o}rnson and L.~Sanguinetti, ``Making cell-free massive {MIMO
  competitive with MMSE} processing and centralized implementation,''
  \emph{accepted in IEEE Trans. Wireless Commun.}, 2020.

\bibitem{Ntontin2020}
K.~Ntontin, M.~Di~Renzo, and F.~Lazarakis, ``On the rate and energy efficiency
  comparison of reconfigurable intelligent surfaces with relays,'' in
  \emph{IEEE 21st International Workshop on Signal Processing Advances in
  Wireless Communications (SPAWC)}.\hskip 1em plus 0.5em minus 0.4em\relax
  IEEE, 2020, pp. 1--5.

\bibitem{Wang2020}
Z.~Wang, L.~Liu, and S.~Cui, ``Channel estimation for intelligent reflecting
  surface assisted multiuser communications: {Framework}, algorithms, and
  analysis,'' \emph{IEEE Trans. Wireless Commun.}, vol.~19, no.~10, pp.
  6607--6620, 2020.

\bibitem{You2020}
C.~You, B.~Zheng, and R.~Zhang, ``Channel estimation and passive beamforming
  for intelligent reflecting surface: {Discrete} phase shift and progressive
  refinement,'' \emph{IEEE J. Sel. Areas Commun.}, vol.~38, no.~11, pp.
  2604--2620, 2020.

\bibitem{You2021}
C.~You and R.~Zhang, ``Wireless communication aided by intelligent reflecting
  surface: {Active} or passive?'' \emph{IEEE Wireless Commun. Lett.}, vol.~10,
  no.~12, pp. 2659--2663, 2021.

\bibitem{Liu2021}
K.~Liu \emph{et~al.}, ``Active reconfigurable intelligent surface:
  {Fully}-connected or sub-connected?'' \emph{IEEE Commun. Lett.}, vol.~26,
  no.~1, pp. 167--171, 2021.

\bibitem{Tasci2022}
R.~A. Tasci \emph{et~al.}, ``A new {RIS architecture with a single power
  amplifier: Energy} efficiency and error performance analysis,'' \emph{IEEE
  Access}, vol.~10, pp. 44\,804--44\,815, 2022.

\bibitem{Liu2021a}
Y.~Liu \emph{et~al.}, ``{STAR}: Simultaneous transmission and reflection for
  $360^\circ$ coverage by intelligent surfaces,'' \emph{IEEE Wireless Commun.},
  vol.~28, no.~6, pp. 102--109, 2021.

\bibitem{Xu2021}
J.~Xu \emph{et~al.}, ``{STAR-RISs}: {Simultaneous} transmitting and reflecting
  reconfigurable intelligent surfaces,'' \emph{IEEE Commun. Lett.}, vol.~25,
  no.~9, pp. 3134--3138, 2021.

\bibitem{Mu2021}
X.~Mu \emph{et~al.}, ``Simultaneously transmitting and reflecting {(STAR) RIS}
  aided wireless communications,'' \emph{IEEE Trans. Wireless Commun.},
  vol.~21, no.~5, pp. 3083--3098, 2021.

\bibitem{Papazafeiropoulos2023a}
A.~Papazafeiropoulos \emph{et~al.}, ``Achievable rate of a {STAR-RIS} assisted
  massive {MIMO} system under spatially-correlated channels,'' \emph{IEEE
  Trans. Wireless Commun.}, pp. 1--1, 2023.

\bibitem{Papazafeiropoulos2023b}
A.~Papazafeiropoulos, P.~Kourtessis, and S.~Chatzinotas, ``{Max-Min SINR}
  analysis of {STAR-RIS} assisted massive {MIMO} systems with hardware
  impairments,'' \emph{IEEE Tran Wireless Commun}, pp. 1--1, 2023.

\bibitem{Papazafeiropoulos2023c}
A.~Papazafeiropoulos \emph{et~al.}, ``{STAR-RIS} assisted cell-free massive
  {MIMO} system under spatially-correlated channels,'' \emph{IEEE Trans. Veh.
  Tech.}, pp. 1--16, 2023.

\bibitem{Zhang2020b}
S.~Zhang \emph{et~al.}, ``Beyond intelligent reflecting surfaces:
  {Reflective}-transmissive metasurface aided communications for
  full-dimensional coverage extension,'' \emph{IEEE Trans. Veh. Tech.},
  vol.~69, no.~11, pp. 13\,905--13\,909, 2020.

\bibitem{Zeng2021}
S.~Zeng \emph{et~al.}, ``Reconfigurable intelligent surfaces in {6G:
  Reflective}, transmissive, or both?'' \emph{IEEE Commun. Letters}, vol.~25,
  no.~6, pp. 2063--2067, 2021.

\bibitem{Zhang2022a}
S.~Zhang \emph{et~al.}, ``Intelligent omni-surfaces: {Ubiquitous} wireless
  transmission by reflective-refractive metasurfaces,'' \emph{IEEE Trans.
  Wireless Commun.}, vol.~21, no.~1, pp. 219--233, 2022.

\bibitem{Liu2022}
Y.~Liu \emph{et~al.}, ``Spectral efficiency maximization for double-faced
  active reconfigurable intelligent surface,'' \emph{IEEE Trans. Signal
  Process.}, vol.~70, pp. 5397--5412, 2022.

\bibitem{Ma2022}
Y.~Ma \emph{et~al.}, ``Optimization for reflection and transmission
  dual-functional active {RIS}-assisted systems,'' \emph{arXiv preprint
  arXiv:2209.01986}, 2022.

\bibitem{Bousquet2012}
J.-F. Bousquet, S.~Magierowski, and G.~G. Messier, ``A 4-{GHz} active scatterer
  in 130-nm {CMOS} for phase sweep amplify-and-forward,'' \emph{IEEE Trans.
  Circuits and Systems I: Regular Papers}, vol.~59, no.~3, pp. 529--540, 2012.

\bibitem{Zhi2022}
K.~Zhi \emph{et~al.}, ``Is {RIS}-aided massive {MIMO} promising with {ZF}
  detectors and imperfect {CSI}?'' \emph{IEEE J. Sel. Areas Commun.}, vol.~40,
  no.~10, pp. 3010--3026, 2022.

\bibitem{Hoydis2013}
J.~Hoydis, S.~ten Brink, and M.~Debbah, ``Massive {MIMO} in the {UL/DL} of
  cellular networks: How many antennas do we need?'' \emph{IEEE J. Select.
  Areas Commun.}, vol.~31, no.~2, pp. 160--171, February 2013.

\bibitem{Bjoernson2020}
E.~{Bj{\"o}rnson} and L.~{Sanguinetti}, ``Rayleigh fading modeling and channel
  hardening for reconfigurable intelligent surfaces,'' \emph{IEEE Wireless
  Commun. Lett.}, vol.~10, no.~4, pp. 830--834, 2021.

\bibitem{Marzetta2016}
T.~L. Marzetta \emph{et~al.}, \emph{Fundamentals of Massive MIMO}.\hskip 1em
  plus 0.5em minus 0.4em\relax Cambridge University Press, 2016.

\bibitem{Ngo2013}
H.~Q. Ngo, E.~Larsson, and T.~Marzetta, ``Energy and spectral efficiency of
  very large multiuser {MIMO} systems,'' \emph{IEEE Trans. Commun.}, vol.~61,
  no.~4, pp. 1436--1449, April 2013.

\bibitem{Nadeem2020}
Q.~{Nadeem} \emph{et~al.}, ``Intelligent reflecting surface-assisted multi-user
  {MISO Communication: Channel} estimation and beamforming design,'' \emph{IEEE
  Open J. Commun. Soc.}, vol.~1, pp. 661--680, 2020.

\bibitem{Kay}
S.~M. Kay, \emph{Fundamentals of statistical signal processing: Estimation
  theory}.\hskip 1em plus 0.5em minus 0.4em\relax Upper Saddle River: Prentice
  Hall PTR, 1993.

\bibitem{massivemimobook}
E.~Bj\"{o}rnson, J.~Hoydis, and L.~Sanguinetti, ``Massive {MIMO} networks:
  {Spectral}, energy, and hardware efficiency,'' \emph{Foundations and
  Trends{\textregistered} in Signal Processing}, vol.~11, no. 3-4, pp.
  154--655, 2017.

\bibitem{Bjoernson2017}
E.~Bj{\"o}rnson \emph{et~al.}, ``Massive {MIMO} networks: Spectral, energy, and
  hardware efficiency,'' \emph{Foundations and Trends{\textregistered} in
  Signal Processing}, vol.~11, no. 3-4, pp. 154--655, 2017.

\bibitem{Kong2015}
C.~Kong \emph{et~al.}, ``Sum-rate and power scaling of massive {MIMO} systems
  with channel aging,'' \emph{IEEE Trans. Commun.}, vol.~63, no.~12, pp.
  4879--4893.

\bibitem{Bertsekas1999}
D.~Bertsekas, \emph{Nonlinear Programming}, 2nd~ed., M.~A. Scientific, Ed.,
  1999.

\bibitem{Absil2008}
P.-A. Absil, R.~Mahony, and R.~Sepulchre, \emph{Optimization algorithms on
  matrix manifolds}.\hskip 1em plus 0.5em minus 0.4em\relax Princeton
  University Press, 2008.

\bibitem{Chen2021a}
Y.~Chen \emph{et~al.}, ``Exploiting reconfigurable intelligent surfaces in edge
  caching: {Joint} hybrid beamforming and content placement optimization,''
  \emph{IEEE Trans. Wireless Commun.}, vol.~20, no.~12, pp. 7799--7812, 2021.

\bibitem{Bjoernson2021}
E.~Bj{\"o}rnson, {\"O}.~T. Demir, and L.~Sanguinetti, ``A primer on near-field
  beamforming for arrays and reconfigurable intelligent surfaces,'' in
  \emph{2021 55th Asilomar Conference on Signals, Systems, and
  Computers}.\hskip 1em plus 0.5em minus 0.4em\relax IEEE, 2021, pp. 105--112.

\bibitem{Bjoernson2014}
E.~Bj{\"o}rnson \emph{et~al.}, ``Massive {MIMO} systems with non-ideal
  hardware: {E}nergy efficiency, estimation, and capacity limits,'' \emph{IEEE
  Trans. Inf. Theory}, vol.~60, no.~11, pp. 7112--7139, 2014.

\bibitem{hjorungnes:2011}
A.~Hj{\o}rungnes, \emph{Complex-Valued Matrix Derivatives: With Applications in
  Signal Processing and Communications}.\hskip 1em plus 0.5em minus 0.4em\relax
  Cambridge University Press.

\end{thebibliography}

\end{document}